\documentclass[10pt,journal,twoside,web]{ieeecolor}
\usepackage{tmi}
\usepackage{cite}
\usepackage{amsmath,amssymb,amsfonts}
\usepackage{algorithm2e}

\usepackage{graphicx}
\usepackage{textcomp}
\newtheorem{theorem}{Theorem}
\def\BibTeX{{\rm B\kern-.05em{\sc i\kern-.025em b}\kern-.08em
    T\kern-.1667em\lower.7ex\hbox{E}\kern-.125emX}}
\begin{document}
\title{Texture Matching GAN for CT Image Enhancement}
\author{Madhuri Nagare, Gregery T. Buzzard, \IEEEmembership{Senior Member, IEEE,} and  Charles A. Bouman, \IEEEmembership{Fellow, IEEE}
\thanks{This paragraph of the first footnote will contain the date on which
you submitted your paper for reviews. M.~Nagare and C.A.~Bouman were partially supported by GE Healthcare, Waukesha, WI, USA.  C.A.~Bouman was partially supported by the Showalter Trust.  G.T.~Buzzard was partially supported by NSF CCF-1763896.}
\thanks{M.~Nagare and C.A.~Bouman are with the 
Elmore Family School of Electrical and Computer Engineering, Purdue University, West
Lafayette, IN, USA (e-mail: mnagare@purdue.edu, bouman@purdue.edu).}
\thanks{G.T.~Buzzard
is with the Department of Mathematics, Purdue University, West Lafayette,
IN, USA (email: buzzard@purdue.edu).}}

\maketitle

\begin{abstract}
Deep neural networks (DNN) are commonly used to denoise and sharpen X-ray computed tomography (CT) images with the goal of reducing patient X-ray dosage while maintaining reconstruction quality.
However, naive application of DNN-based methods can result in image texture that is undesirable in clinical applications.
Alternatively, generative adversarial network (GAN) based methods can produce appropriate texture, but naive application of GANs can introduce inaccurate or even unreal image detail.
In this paper, we propose a texture matching generative adversarial network (TMGAN) that enhances CT images while generating an image texture that can be matched to a target texture.
We use parallel generators to separate anatomical features from the generated texture, which allows the GAN to be trained to match the desired texture without directly affecting the underlying CT image.
We demonstrate that TMGAN generates enhanced image quality while also producing image texture that is desirable for clinical application. 

\end{abstract} 

\section{Introduction}
\label{sec:introduction}
X-ray computed tomography (CT) is one of the most widely used 3D medical imaging modalities, with recent progress on reconstruction methods resulting in reduced noise and artifacts while improving resolution and quality \cite{textbookJiang}.
In particular, noise reduction and image sharpening methods can be used to reduce X-ray dosage while maintaining image quality.   
However, the true measure of quality for a medical CT reconstruction method is its ability to improve diagnostic accuracy.  

Radiologists regard the texture of CT images after enhancement as critically important for diagnosis \cite{waite2019analysis}. 
In fact, high quality texture provides important visual cues in decision making for radiologists \cite{texture_medical_diag}.
Since radiologists are typically familiar with the noise texture of filtered back projection (FBP) \cite{fdk}, this texture is often described as desirable \cite{DLIR_testm,pan2020impact}.
Quantitatively, \cite{racine2020task} found that reducing noise while maintaining a texture like that of FBP led to better lesion detection than noise reduction that changed the texture.   

Approaches to control CT image texture include  
\cite{gatys2016image, gatys2015texture}, which synthesize high quality texture by matching the statistics of generated texture to a target.
However, employing this method while denoising or sharpening a CT image requires separation of true anatomy from texture.  This separation can be done using morphological component analysis (MCA) \cite{starck2005image} as in \cite{zeng2020magnetic, cheng2016image}.
However, MCA requires dictionary learning for both the object and texture components, which is computationally expensive.

From \cite{kulathilake2022review}, adaptive methods to preserve textural information include patch-based approaches using spatial similarity \cite{nlm, NLMpaper, bm3d}, which employ parameters to control image smoothness.  Alternatively,
\cite{MBIRGE} used a tuned a prior distribution in an iterative reconstruction \cite{beister2012iterative} to produce desirable texture. 
However, these methods tend to be computationally expensive. 

Deep neural networks (DNNs) are currently among the most popular methods for CT image enhancement \cite{Kang2017, SRCNN2}.~\cite{BRpaper} and~\cite{NPSFPaper} proposed a novel loss  function and a method to create training pairs, respectively, so that a DNN can be trained to preserve the texture.
Another approach to improving texture using DNNs is to use a generative adversarial network (GAN) \cite{Wolterink2017}, with the generator output (i.e., the denoised or sharpened image) as an input to the discriminator \cite{Yang2018a, SACNN, GANsuperRes}.
While these approaches can produce more realistic texture, they do not allow optimization to produce a particular desirable target texture.
Also, in GAN based methods, since the underlying CT image is not separated from the texture, the discriminator could encourage the addition of inaccurate or even unreal image detail known as hallucinations~\cite{gan_hallc}.
Xian \textit{et al.} \cite{xian2018texturegan} used a conditional GAN to produce a target texture in natural images when an image sketch is provided. 
However, this method cannot be applied directly to our problem since sketches of anatomy are not generally available.

In this paper, we propose the Texture Matching GAN (TMGAN), which denoises and/or sharpens CT images while simultaneously matching the generated texture to a distribution of target textures.
The methods of TMGAN build on our earlier research presented in~\cite{BRpaper}.
A novel aspect of TMGAN is that it separates the texture from the underlying clean CT image by adding two independent noise samples to the same ground truth image and processing them with a Siamese network \cite{siamese} (generator) to produce two conditionally independent estimates. 
We take the difference of these two estimates to separate the texture component from the underlying clean CT image.
This allows the GAN to be trained without the risk of generating false image detail or hallucinations.

\begin{figure*}[t]
     \centering
     \includegraphics[scale=0.6]{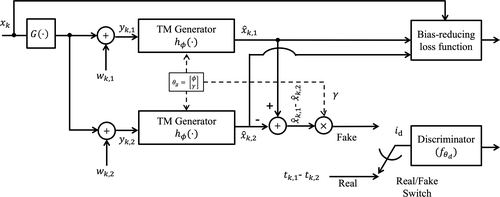}
     \caption{Network architecture for Texture Matching GAN (TMGAN).}
     \label{fig_network}
 \end{figure*}

Our main contributions is a TMGAN architecture that:
\begin{itemize}
\item Denoises or sharpens a CT image while generating image texture that matches a desired target texture;
\item Separates noise texture from the underlying clean CT image by subtracting two conditionally independent estimates with the same ground truth;
\item Uses a novel bias reduction method to reduce bias in the estimated image;
\end{itemize}
We demonstrate the effectiveness of the TMGAN approach on simulated and experimentally measured CT data and show both quantitatively and qualitatively that TMGAN yields better texture quality than existing approaches.

\section{Problem Formulation}\label{problem}
Let $Y$ be an observed image from which we aim to recover the true image $X$. We assume the following forward model,
\begin{align}\label{formodel}
    Y = G(X) + W,
\end{align}
where $W$ is noise and $G(\cdot)$ models other deformations in $Y$; for example, in sharpening applications, $G(\cdot)$ is a blurring function. We also assume that the noise $W$ is zero mean and independent of $X$, so that $\mathbb{E}[W|X] = 0.$  We then seek to estimate $X$ as $\hat{X} = h(Y)$.  

To model texture, we define the estimate $\hat{X}$ as a sum of three components,
\begin{align}
    \hat{X} = X + B_X + \delta_X,
\end{align}
where $B_X = \mathbb{E}[\hat X| X]-X$ is the bias in the estimate and $\delta_X = \hat X - \mathbb{E}[\hat X| X]$ is the estimation noise, or texture, with $\mathbb{E}[\delta_X| X]=0$. 
Note that if $X$ is known, then $B_X$ is deterministic, whereas $\delta_X$ is still a random variable and a function of $X$.

Our goal is to train a conditional GAN whose generator estimates $X$ from $Y$ and that produces estimation noise (or texture) matching the distribution of a specified distribution of textures.  That is, the generator gives $\hat{X} = h(Y)$, and the resulting $\delta_X$ should be distributed to match user-provided independent samples of the target noise texture, $T$.  
However, it is difficult to design a loss function for this task because we observe the ground truth, $X$, and the estimate, $\hat{X}$, but never directly observe the estimation noise, $\delta_X$.

The key to training TMGAN is to add two independent noise instances $W_1$ and $W_2$ using a single ground truth $X$ in \eqref{formodel} to form $Y_1$ and $Y_2$.
These are used to give estimated images $\hat{X}_1=h(Y_1)$ and $\hat{X}_2=h(Y_2)$, 
which satisfy
\begin{align}\label{diffeq}
    \hat X_1 - \hat X_2 &=  (X + B_X + \delta_{X_1}) - (X + B_X +     \delta_{X_2}) \nonumber \\
        &=  \delta_{X_1} - \delta_{X_2},
\end{align}
where $\delta_{X_1}$ and $\delta_{X_2}$ are conditionally independent and identically distributed (i.i.d.) given $X$. 
Since our goal is to match $\delta_{X_1}$ to the distribution of $T$, \eqref{diffeq} implies that $\hat X_1 - \hat X_2$ should match the distribution of $T_1 - T_2$, where $T_1$ and $T_2$ are two i.i.d. samples of the target texture, and this is easy to achieve in a GAN framework.  

The question remaining is whether matching $\delta_{X_1} - \delta_{X_2}$ to $T_1 - T_2$ will match the distributions of $\delta_X$ and $T$.  For Gaussian $T$, we answer in the affirmative in the following theorem.  

\begin{theorem}
Let $\delta_1 - \delta_2 \sim \mathcal{N}(0, 2\sigma ^2)$ and $\delta_1, \delta_2$ be real valued i.i.d. random variables with a distribution that is symmetric about $0$.
Then $\delta_1 \sim \mathcal{N}(0, \sigma^2).$ 
\end{theorem}
\begin{proof}
Let $Z = \delta_1 - \delta_2$. Since $Z$, $\delta_1$, and $\delta_2$ are real valued, all of their characteristic functions exist.  
Since $Z \sim \mathcal{N} (0, 2 \sigma^2)$, its characteristic function is given by
\begin{align*}
    \phi_Z(t) = \mathbb{E}[ e^{jt Z} ] = e^{-\sigma^2t^2} \ .
\end{align*}
Since $\delta_1$ and $\delta_2$ are independent, we have that
\begin{align*}
\phi_{Z}(t) &= \mathbb{E}[e^{jt(\delta_1 - \delta_2)}] \\
    &= \mathbb{E}[e^{jt\delta_1}e^{-jt\delta_2}] \\
    &= \mathbb{E}[e^{jt\delta_1}]\mathbb{E}[e^{-jt\delta_2}] 
\end{align*}
Since $\delta_1$ and $\delta_2$ are symmetric about $0$, we can remove the negative sign in the final expected value.  Since $\delta_1$ and $\delta_2$ are i.i.d., the 2 expected values are the same, hence 
\begin{align*}
\phi_{Z}(t) 
    &= \mathbb{E}[e^{jt\delta_1}]^2 \ .
\end{align*}
Taking square roots yields
\begin{align*}
\mathbb{E}[e^{jt\delta_1}]  
    &= \pm e^{-\frac{1}{2} \sigma^2t^2}
\end{align*}
Since the characteristic function is always continuous, the choice of $\pm$ is independent of $t$.  Since the left hand side is 1 when $t=0$, we see that the characteristic function of $\delta_1$ is $e^{-\frac{1}{2} \sigma^2t^2}$.  By uniqueness of the characteristic function, this means that $\delta_1 \sim \mathcal{N} (0, \sigma^2)$, hence likewise for $\delta_2$.
\end{proof}

\section{Texture Matching GAN}\label{tmganarch}

\subsection{TMGAN Architecture}

Fig.~\ref{fig_network} shows the network architecture for TMGAN, where lower-case letters denote samples of the aforementioned random variables.
The central component is the TM generator, $h_\phi(\cdot)$, a neural network parameterized by the vector $\phi$, which uses noisy and distorted input, $y$, to estimate the original image as $\hat{x} = h_\phi(y)$. 
As described below, we optimize $h_\phi$ not only to minimize mean squared error (MSE) but also to produce a texture that statistically matches the provided training or target texture samples.
Since the training texture samples may have a different amplitude than the estimation texture, we use a parameter, $\gamma$, to account for the potential difference in scaling.
This $\gamma$ parameter can be learned or set manually.

For the $k^{th}$ ground truth image, we first apply the deformation $G(\cdot)$ (Gaussian blur for sharpening and the identity for denoising) and then generate two conditionally independent inputs $y_{k,1}$ and $y_{k,2}$ by adding independent noise samples $w_{k, 1}$ and $w_{k, 2}$.
From these two inputs, the TM generator produces two estimates, $\hat{x}_{k,1} = h_\phi(y_{k,1})$ and $\hat{x}_{k,2} = h_\phi(y_{k,2})$. 
Following \eqref{diffeq}, we take the scaled difference $\gamma (\hat{x}_{k,1} -\hat{x}_{k,2})$ to get a sample of the fake texture difference, $\gamma (\delta_{x1} - \delta_{x2})$. 
We also generate samples of the real texture difference $t_{k,1} -t_{k,2}$ using sample images of the target texture. 

\subsection{TMGAN Training}
To train the TM generator, we need to promote accurate image estimation and match generated texture differences to target texture differences.   
To match texture differences, we use a GAN framework \cite{GANorig} in which the discriminator network is trained to differentiate scaled fake texture differences from real texture differences. 
We model the discriminator as a function $f_{\theta_d}(\cdot)\in (0,1)$, with parameters $\theta_d$, and we interpret $f_{\theta_d}$ as the probability that a texture difference sample is real.
\RestyleAlgo{ruled}
\begin{algorithm}[t!]
\SetKwComment{Comment}{/* }{ */}
\SetKwInOut{Input}{Input}%

\caption{Training Pseudocode for TMGAN}\label{training}
\Input{ $N:$ total number of generator updates,\newline
 $T_d:$ threshold for the discriminator loss,\newline
$N_d:$ maximum number of discriminator updates per generator update\newline}

$\theta_g, \theta_d \gets$ Initialize network parameters.

\For{$n \gets 0 \text{ to } N$}{
$n_d \gets 0$

\While{$d(\theta_g, \theta_d) > T_d \textbf{ AND } n_d < N_d$}{

 $\theta_d \gets$ Update $\theta_d$ with one iteration of Adam optimizer to minimize $d(\theta_g, \theta_d)$

 $n_d \gets n_d+1$}
 $\theta_g \gets$ Update $\theta_g$ with one iteration of Adam optimizer to minimize $g(\theta_g, \theta_d)$}
\end{algorithm}

Using this notation, the discriminator is trained by minimizing the Binary Cross-Entropy (BiCE) \cite{GANorig} loss function with respect to $\theta_d$:
\begin{align}\label{disccost}
    d(\theta_g, \theta_d) = & -\frac{1}{K}  \sum_{k=1}^K  \big[
    \log{\big(f_{\theta_d}(t_{k,1} - t_{k,2}) \big)} \nonumber\\
    &+ \log{\big(1-f_{\theta_d}(\gamma ( h_\phi(y_{k,1}) - h_\phi(y_{k,2}) )\big) \big]} \ ,
\end{align}
where  $\theta_g = [\phi, \gamma]$ and $K$ is the number of training samples. 

Unlike previous GAN based methods for CT image enhancement, the TMGAN discriminator works only on the texture part, which avoids the risk of possible addition of fake detail, known
as hallucinations, in the enhanced images.

Since our goal is to match ground truth while maximizing texture quality, we need two loss terms for the generator, one to promote desirable textures and a second to minimize MSE.  
The MSE term incorporates the bias-reducing loss function of \cite{BRpaper}, which demonstrated that the bias-reducing loss function yields better structural detail in denoised images than was obtained with the standard MSE loss. 

Using this strategy, the TMGAN generator loss function is 
\begin{align}\label{gencost}
    g(\theta_g, \theta_d) =& \frac{1}{K} \sum_{k=1}^{K} 
    \bigg[ -\lambda\log{ \big( \gamma ( h_\phi(y_{k,1}) - h_\phi(y_{k,2}) ) \big) }\nonumber\\&+\frac{1}{2\sigma^2}\bigg(\|\hat{z}_{k,1} - x_k\|^2+\|\hat{z}_{k,2} - x_k\|^2\bigg) \bigg], 
\end{align}
where 
\begin{align*}
{\hat z}_{k,1} &= \alpha \hat{x}_{k,1} + (1-\alpha) \hat{x}_{k,2}  \\
{\hat z}_{k,2} &= (1-\alpha) \hat{x}_{k,1} + \alpha \hat{x}_{k,2} \ 
\end{align*}
form the basis of the bias-reducing loss, $\lambda$ weights texture loss versus fit to data, and $\sigma$ is roughly the standard deviation in estimating $x_k$. 
We set $\alpha=0.5$ for maximum bias reduction for denoising \cite{BRpaper} and $\alpha=1.0$ to reduce aliasing in the sharpening case.
Note that each of the two branches of the TM Generator network share the same parameters, so this can be treated as a Siamese network \cite{siamese} for training.

\begin{figure}[t]
    \centering
    \includegraphics[scale=0.37]{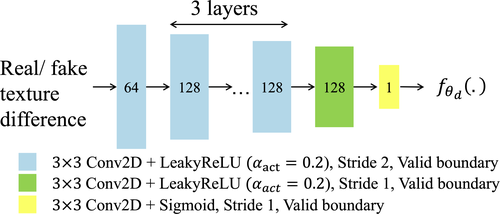}
    \caption{Discriminator architecture for TMGAN}
    \label{figtmgandisc}
\end{figure}

  \begin{table}
    \centering
        \caption{Summary of parameters for TMGAN}
   \begin{tabularx}{\linewidth} { 
  | m{1.3cm}
  | m{3.085cm}  | m{0.5cm}
  | m{0.5cm} | }
    \hline
        \multicolumn{1}{|c|}{\makecell{Parameter\\(range)}}&  \multicolumn{1}{c|}{\makecell{Significance}}& \multicolumn{2}{c|}{\makecell{Values used\\in Experiments}} \\ &&\multicolumn{1}{c|}{\makecell{Denoising}}&\multicolumn{1}{c|}{\makecell{Sharpening}}\\\hline
         \multicolumn{1}{|c|}{\makecell{$\alpha$ \\$(\in [0.5, 1.0])$}} &Smaller $\alpha$ increases detail by reducing bias. See~\eqref{gencost}. &\multicolumn{1}{c|}{\makecell{$0.5$}}&\multicolumn{1}{c|}{\makecell{$1.0$}} \\\hline
         $\lambda \hspace{1ex} (\geq 0)$ & Larger $\lambda$ improves texture match. See~\eqref{gencost}&  \multicolumn{2}{c|}{\makecell{See results}} \\\hline
         $\sigma \hspace{1ex} (> 0)$& Smaller $\sigma$ reduces squared error. See~\eqref{gencost}.
           & \multicolumn{2}{c|}{\makecell{See results}} \\\hline
         \multicolumn{1}{|c|}{\makecell{$\eta$ \\$(\in [0, 1.0])$}}& Fraction of TMGAN in TMGAN-blended image. See~\eqref{blendingeq}. & \multicolumn{1}{c|}{\makecell{$0.3$}}& \multicolumn{1}{c|}{\makecell{$1.0$}}\\\hline
         $T_d \hspace{1ex} (> 0)$& The Discriminator is updated if its loss is greater than $T_d$. See Alg.~\ref{training}. &\multicolumn{1}{c|}{\makecell{0.2}}&\multicolumn{1}{c|}{\makecell{0.2}}\\\hline
         $N_d \hspace{1ex} (\geq 1)$& Maximum number of discriminator updates per generator update. See Alg.~\ref{training}. &\multicolumn{1}{c|}{\makecell{$1$}}&\multicolumn{1}{c|}{\makecell{$5$}}\\\hline
            $N\hspace{1ex}\gg 1$& The total number of generator updates. See Alg.~\ref{training}. &\multicolumn{1}{c|}{\makecell{$\sim 15$}}&\multicolumn{1}{c|}{\makecell{$\sim 15$}} \\\hline
    \end{tabularx}
    \label{tab:para}
\end{table}

Algorithm~\ref{training} shows the pseudocode to train the TM generator.
To avoid common instability and nonconvergence issues encountered in training GANs \cite{DLbook}, we use a training procedure to reduce mode collapse and vanishing gradients. 
An optimal discriminator can help avoid mode collapse \cite{wgan}, so we use multiple discriminator updates between generator updates. 
However, an optimal discriminator might not provide enough gradient to the generator to make progress, so to avoid vanishing gradients, we update discriminator weights only if the loss is greater than a threshold $T_d.$

For inference, we employ the trained $h_\phi(\cdot )$ only.
   \begin{table}[t]
        \centering 
         \caption{Clinical test exams}
        \begin{tabular}{|c|c|c|c|c|c|c|c|}%
        \hline
        \multicolumn{1}{|c|}{\makecell{Exam\\ name}}&
        \multicolumn{1}{c|}{\makecell{Scanned\\ object}}&
        \multicolumn{1}{c|}{\makecell{Focal \\spot size}}&
        \multicolumn{1}{c|}{\makecell{Dosage\\ (kVp/mA)}}&
        \multicolumn{1}{c|}{\makecell{DFOV \\(cm)}}\\\hline
        Exam 1&
        \multicolumn{1}{c|}{\makecell{Clinical body}}&
        XL&
        80/375&
        \multicolumn{1}{c|}{31.1}\\\hline
        
        Exam 2&
        \multicolumn{1}{c|}{\makecell{Clinical body}}&
        small&
        100/220&
        \multicolumn{1}{c|}{49.2}\\\hline
        
        Exam 3&
        \multicolumn{1}{c|}{\makecell{Clinical body}}&
        small&
        120/110&
        \multicolumn{1}{c|}{35.0}\\\hline

        Exam 4&
        \multicolumn{1}{c|}{\makecell{Clinical body}}&
        small&
        120/350&
        \multicolumn{1}{c|}{39.4}\\\hline
        
       Exam 5&
        \multicolumn{1}{c|}{\makecell{Clinical head}}&
        XL&
        120/530&
        \multicolumn{1}{c|}{15}\\\hline

        \end{tabular}
        \label{testdata}
    \end{table}
\begin{figure}[t]
\begin{minipage}{\linewidth} 
  \begin{tabular}{cc}
          \multicolumn{1}{c}{\makecell{Input texture}} & \multicolumn{1}{c}{\makecell{Target texture}}\\
         \begin{subfigure}[b]{2.5cm}
            \centerline{\includegraphics[width=2.0cm]{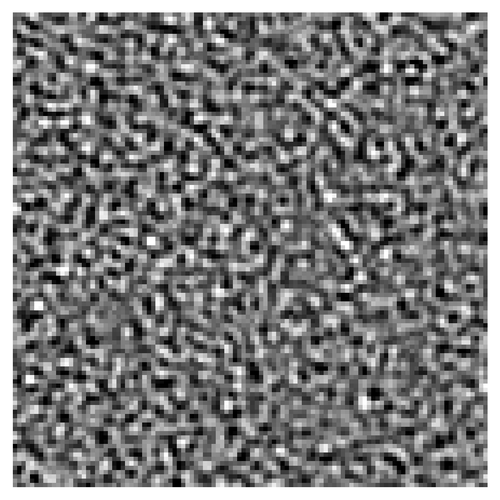}}
            \subcaption{}
            \label{wpinput}
        \end{subfigure}%
     &%
        \begin{subfigure}[b]{2.5cm}
            \centerline{\includegraphics[width=2.0cm]{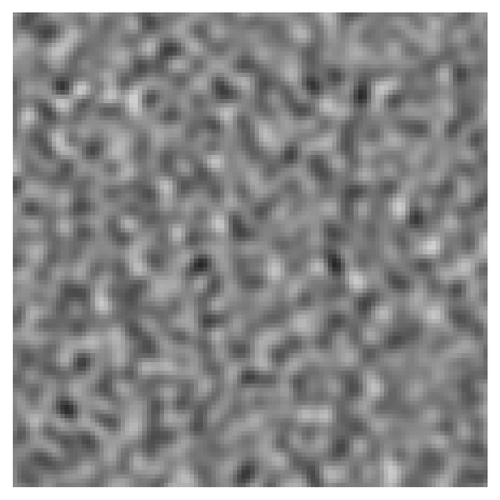}}
            \subcaption{}
            \label{wprealtexture} 
        \end{subfigure}%
  \end{tabular}\\
    \begin{tabular}{ccc}
    \multicolumn{1}{c}{\makecell{TMGAN \\$\lambda=0.0$}}&
    \multicolumn{1}{c}{\makecell{TMGAN\\$\lambda=0.01$}}&
    \multicolumn{1}{c}{\makecell{TMGAN\\$\lambda=0.04$}}\\
        \begin{subfigure}[b]{2.5cm}
            \centerline{\includegraphics[width=2.0cm]{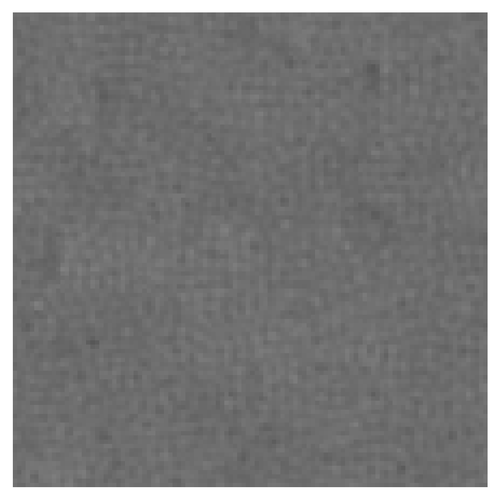}}
            \subcaption{}
            \label{lam0}
        \end{subfigure}%
     &%
        \begin{subfigure}[b]{2.5cm}
            \centerline{\includegraphics[width=2.0cm]{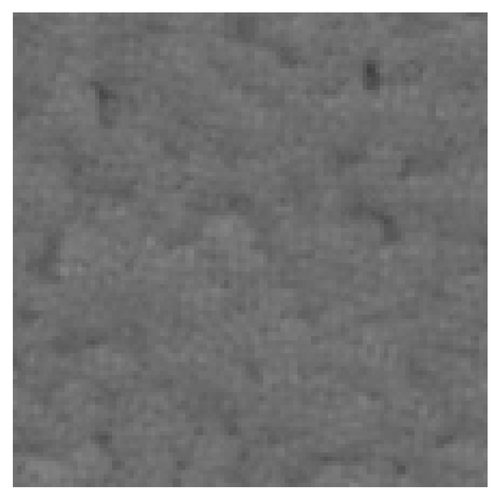}}
            \subcaption{}
            \label{lam001}
        \end{subfigure}%
    &%
         \begin{subfigure}[b]{2.5cm}
             \centerline{\includegraphics[width=2.0cm]{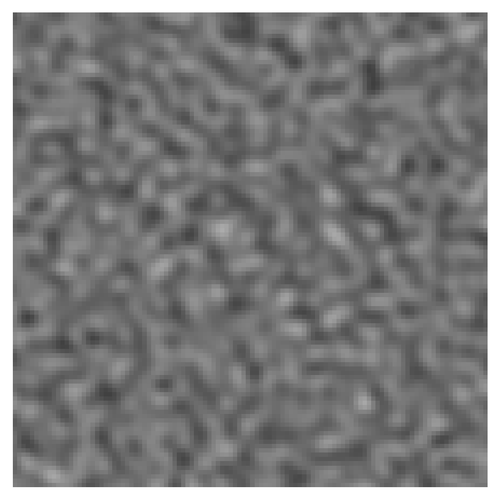}}
             \subcaption{}
             \label{lam004}
         \end{subfigure}%
\end{tabular}
\caption{
Comparison of TMGAN generated textures.
(\protect\subref{wpinput}) Input to TMGAN (water phantom with bone+ filter recon),
(\protect\subref{wprealtexture}) Target texture (water phantom with standard filter recon), 
Result of TMGAN with 
(\protect\subref{lam0}) $\lambda=0$, (\protect\subref{lam001}) $\lambda=0.01$, and (\protect\subref{lam004}) $\lambda=0.04$. 
Notice that texture in results becomes more similar to the target texture as $\lambda$ increases.
}
\label{figWP_example}
\end{minipage}
\end{figure}

\begin{figure}[t]
    \centering
    \includegraphics[width=0.9\linewidth]{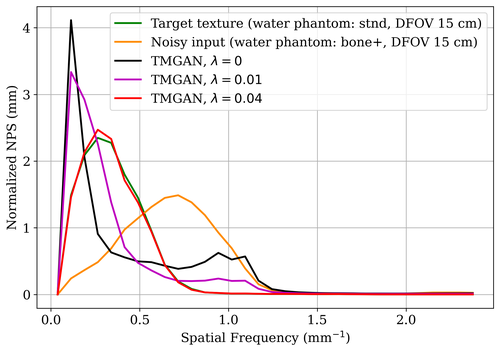}
    \caption{NPS plots for texture images in Fig.~\ref{figWP_example}. Notice that as $\lambda$ increases, the shape of the NPS plots becomes more similar to the NPS of the target texture.}
    \label{fignps1d}
\end{figure}

\subsection{Blending}

To provide fine-grained control over the amount of texture, we use a blending or averaging scheme between TMGAN and standard image estimation.  When $\lambda = 0$ in \eqref{gencost}, TMGAN reduces to the bias reducing network of \cite{BRpaper}, which we call BR-$\alpha$.  
This gives an accurate image estimate but without the target texture. The blended estimate is then
\begin{align}\label{blendingeq}
    \hat{x}^{\text{(TMGAN-blended)}} = \eta h_\phi(y) + (1-\eta)h_{BR-\alpha}(y),
\end{align}
where $y$ is the noisy input, $\eta$ is the blending ratio, and $h_{BR-\alpha}(y)$ has no texture matching term.

When $\eta=1.0$, the blended result is pure TMGAN, and when $\eta=0.0$ it is pure image estimation.
However, for intermediate values of $\eta$, it blends these two results, with larger values of $\eta$ resulting in more  texture, and smaller values resulting in reduced noise and more image detail.

\section{Methods}\label{emethod}

\subsection{TMGAN Implementation}

We train TMGAN separately for the two applications of denoising and sharpening of noisy CT images. 
While the network architecture remains the same for both applications, the training data and test data is different.

For the generator architecture, we used a CNN adopted from \cite{zhang2017beyond} with a single input channel and 17 convolution layers. We modified the discriminator architecture in \cite{isola2017image} to approximately match its strength (the number of trainable parameters) to the generator.
Fig.~\ref{figtmgandisc} shows the discriminator architecture consisting of a series of 2D convolutional layers.

Table \ref{tab:para} lists all the hyperparameters used in the TMGAN algorithm along with their significance, ranges and settings used in the experiments. 
The settings for $\lambda$ and $\sigma$ varied by experiment and are listed with the corresponding results. The values for all the parameters were selected empirically.

\begin{table}
    \centering
        \caption{Standard deviation for results in Fig.~\ref{figWP_example}}
    \begin{tabular}{cc}
    \hline
        Result & Standard deviation in HU  \\\hline
        Fig.~\ref{figWP_example}(\subref{wpinput}): Input texture  & 70.20 \\\hline
         Fig.~\ref{figWP_example}(\subref{wprealtexture}): Target texture & 30.42\\\hline
         Fig.~\ref{figWP_example}(\subref{lam0}): TMGAN $\lambda=0.0$& 6.14 \\\hline
         Fig.~\ref{figWP_example}(\subref{lam001}): TMGAN $\lambda=0.01$& 8.32\\\hline
         Fig.~\ref{figWP_example}(\subref{lam004}): TMGAN $\lambda=0.04$& 20.55\\\hline
    \end{tabular}
    \label{stdfig3}
\end{table}
\begin{figure}[t]
    \centering
    \includegraphics[width=0.9\linewidth]{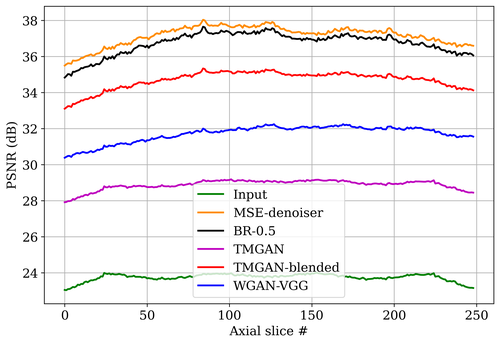}
    \caption{Comparison of PSNR for axial slices of a synthetic exam. TMGAN trained with $\lambda = 0.4, \sigma=7.8$ HU. TMGAN-blended preserves texture with a small reduction in PSNR.}
    \label{figpsnr}
\end{figure}

\subsection{Datasets}
All the scans used in training and evaluating TMGAN were acquired using a GE Revolution CT scanner (GE HealthCare, WI, USA)\footnote{We thank GE HealthCare for collecting the datasets}. The scans are reconstructed to a slice thickness of 0.625 mm and dimension $512 \times 512.$ The standard (stnd) reconstruction kernel option available on the scanner and 40 cm 
DFOV (Display-Field-of-View) are used unless specified otherwise. The scans are reconstructed in HU units, and we added an offset of 1000 to all the images while training and testing so that air is 0. 

\subsubsection{Training and Validation Data}
Ground truth images for denoiser training were generated by reconstructing 10 clinical scans with X-ray tube voltage and current varying from scan to scan in the range of 100-120 kVp and 445-1080 mA, respectively.
The scans  were reconstructed using the GE's TrueFidelity DLIR option \cite{GEwhitepaper}. 

Four ground truth images for sharpener training were obtained by averaging repeated scans of two distinct head phantoms in order to reduce noise.
Each scan was acquired with a small focal spot size at 120 kVp/ 320 mA and reconstructed with a bone+ kernel\cite{NPSFPaper} to a DFOV of 15 cm.

Noise and texture samples for training were generated by removing the mean from scans of 6 water phantoms, obtained with a tube voltage of 120 kVp, current of 350 to 380 mA.

For denoising, the deformation operator, $G(.)$, was simply the identity operator.
This was followed by the addition of two independent noise samples from the water phantoms to generate the two conditionally independent noisy input image samples, $y_{k,1}$ and $y_{k,2}$.
For sharpening, the deformation operator, $G(.)$, was the application of a Gaussian filter of standard deviation $0.244$ mm, $0.244$ mm, $0.344$ mm in $x, y,$ and $z$ directions, respectively. 
This was again followed by the same process to form $y_{k,1}$ and $y_{k,2}$.

In both tasks, the GAN training and validation data was produced by breaking slices into $128 \times 128 \times 1$ patches, with the patches randomly partitioned as 97\% for training and 3\% for validation. 
All DNN trainings were performed using the Adam optimizer \cite{ADAM} with a learning rate of ${3\times10^{-5}}$ for the generator and ${3\times10^{-6}}$ for the discriminator and a mini-batch size of 32. 
The network was implemented in TensorFlow \cite{tensorflow2015-whitepaper} and trained with an NVIDIA Tesla V100 GPU.

\begin{table}[t]
    \centering
       \caption{PSNR and SSIM for 9 synthetic test exams}
    \begin{tabular}{m{2.3cm}m{2.35cm}m{2.35cm}}
    \hline
         \textbf{Method} & \textbf{PSNR} (mean $\pm$ std)  & \textbf{SSIM} (mean $\pm$ std) \\\hline
        Input & $23.74 \pm 0.01$ & $0.70 \pm 0.07$\\\hline
        MSE-denoiser& $37.20 \pm 0.77$ & $0.81 \pm 0.09$\\\hline
        BR-$0.5$ & $36.85 \pm 0.81$ & $0.82 \pm 0.08$\\\hline
       WGAN-VGG & $30.75 \pm 0.56$ & $0.77 \pm 0.09$\\\hline
        TMGAN & $28.98 \pm 0.18$ & $0.74 \pm 0.08$\\\hline
       TMGAN-blended & $34.87 \pm 0.58$ & $0.79 \pm 0.09$\\\hline
    \end{tabular}
    \label{tab:quanteval}
\end{table}

\begin{figure}[t]
    \centering
    \includegraphics[width=0.9\linewidth]{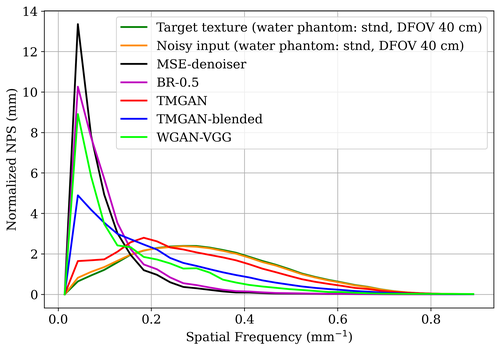}
    \caption{NPS plots of water phantom reconstructions for denoising algorithms. TMGAN (trained with $\lambda = 0.4, \sigma=7.8$ HU) produces closest match with the NPS of target texture. \mbox{TMGAN-blended} is second best with NPS slightly skewed towards the origin while preserving higher frequencies too.}
    \label{fignpsdenoised}
\end{figure}
\ifx\addfigs\undefined
\begin{figure*}[t!]
\begin{minipage}{\textwidth} 
    \begin{tabular}{c@{}c@{}c@{}c@{}c@{}c}
    Noisy input & MSE-denoiser & BR-$0.5$&WGAN-VGG&TMGAN&TMGAN-blended\\
        \begin{subfigure}[b]{3cm}
            \centerline{\includegraphics[width=3cm]{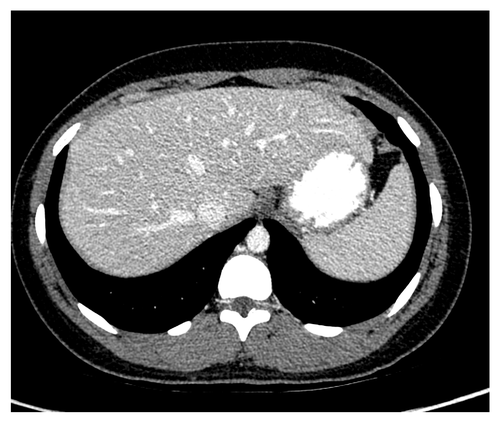}}
            \subcaption{}
            \label{in1}
        \end{subfigure}
     &%
        \begin{subfigure}[b]{3cm}
            \centerline{\includegraphics[width=3cm]{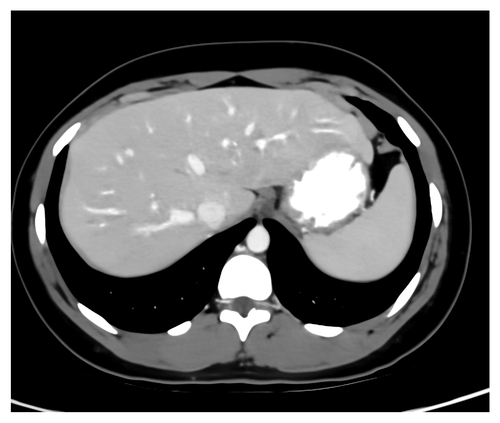}}
            \subcaption{}
            \label{mse}
        \end{subfigure}
    &%
         \begin{subfigure}[b]{3cm}
             \centerline{\includegraphics[width=3cm]{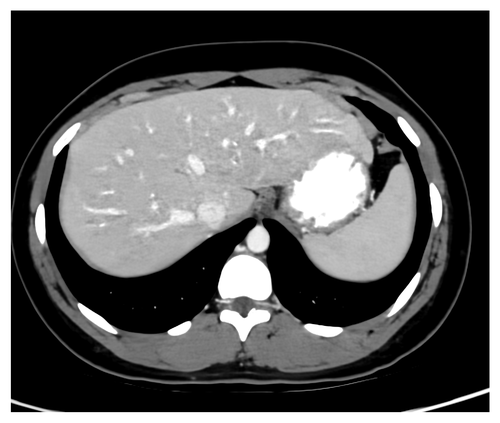}}
             \subcaption{}
             \label{br1}
         \end{subfigure}
     &%
        \begin{subfigure}[b]{3cm}
            \centerline{\includegraphics[width=3cm]{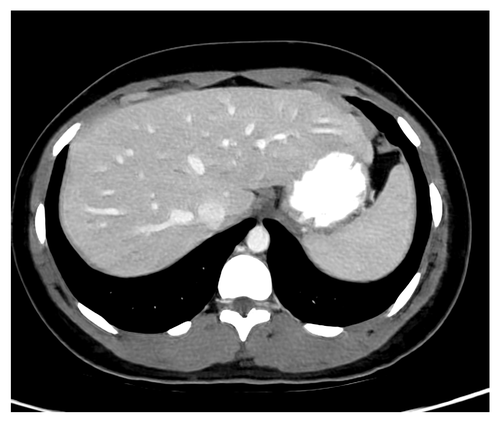}}
            \subcaption{}
            \label{wgan1}
        \end{subfigure}&
        \begin{subfigure}[b]{3cm}
            \centerline{\includegraphics[width=3cm]{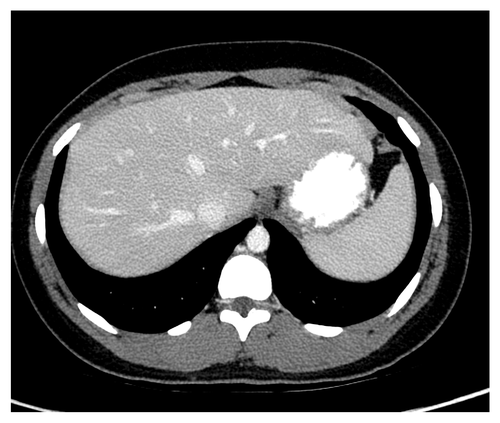}}
            \subcaption{}
            \label{tmgan1}
        \end{subfigure}&
        \begin{subfigure}[b]{3cm}
            \centerline{\includegraphics[width=3cm]{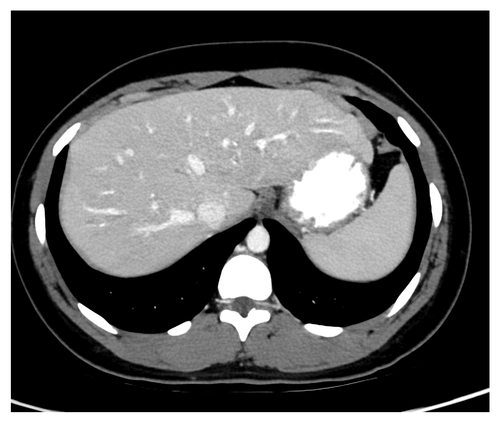}}
            \subcaption{}
            \label{tmganbl1}
        \end{subfigure}\\
                \begin{subfigure}[b]{3cm}
            \centerline{\includegraphics[width=3cm]{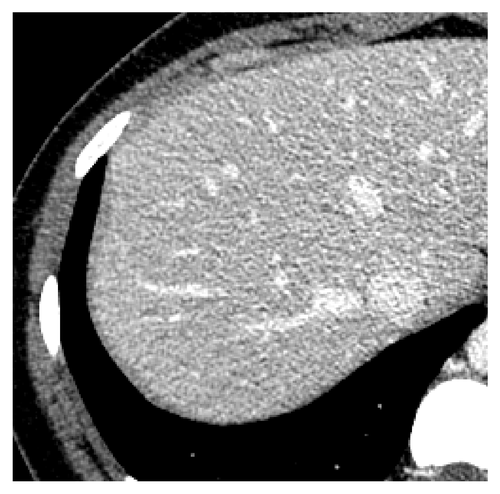}}
            \subcaption{}
            \label{in2}
        \end{subfigure}
     &%
        \begin{subfigure}[b]{3cm}
            \centerline{\includegraphics[width=3cm]{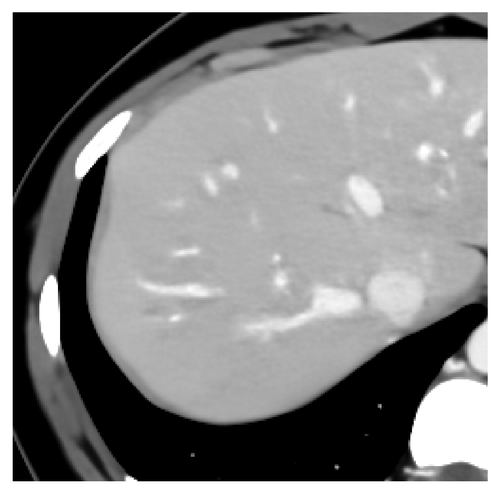}}
            \subcaption{}
            \label{mse2}
        \end{subfigure}
    &%
         \begin{subfigure}[b]{3cm}
             \centerline{\includegraphics[width=3cm]{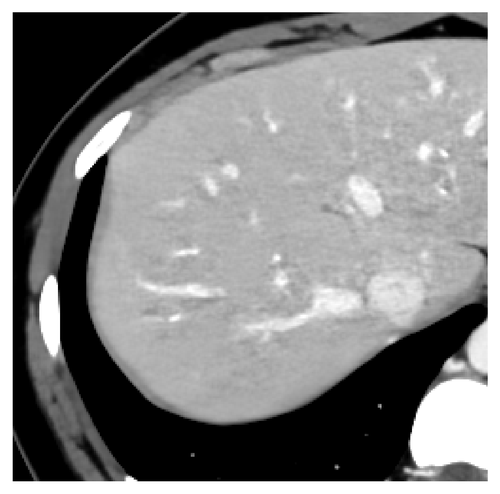}}
             \subcaption{}
             \label{br2}
         \end{subfigure}
     &%
        \begin{subfigure}[b]{3cm}
            \centerline{\includegraphics[width=3cm]{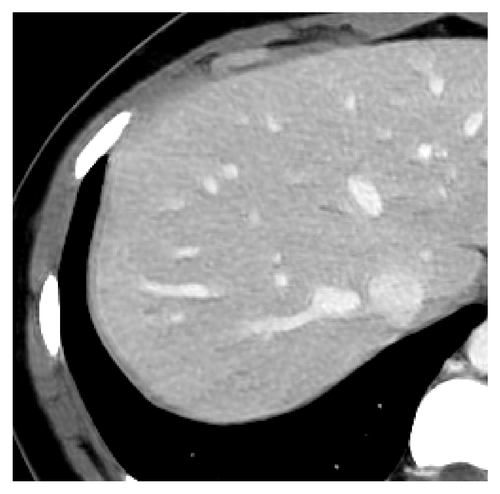}}
            \subcaption{}
            \label{wgan2}
        \end{subfigure}&
        \begin{subfigure}[b]{3cm}
            \centerline{\includegraphics[width=3cm]{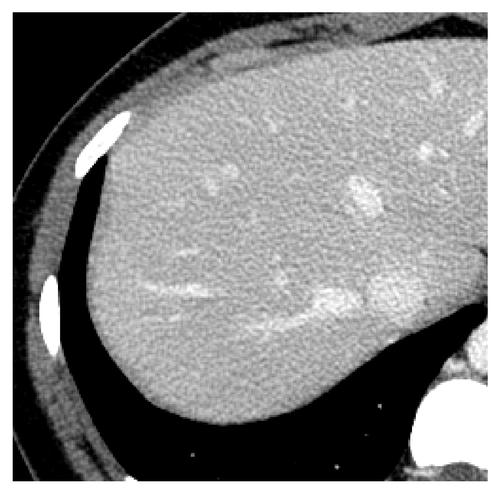}}
            \subcaption{}
            \label{tmgan2}
        \end{subfigure}&
        \begin{subfigure}[b]{3cm}
            \centerline{\includegraphics[width=3cm]{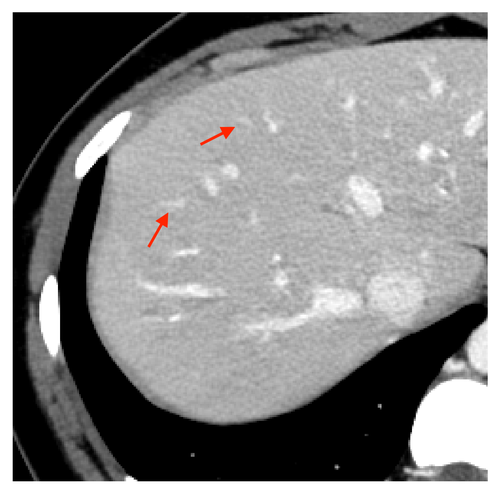}}
            \subcaption{}
            \label{tmganbl2}
        \end{subfigure}
\end{tabular}
\caption{Comparison of slice 1 of denoised results for Exam 1, a low-dose clinical scan. First row shows full slice. Second row shows zoomed ROI. TMGAN trained with $\lambda = 0.4, \sigma=7.8$ HU. Display window is [-125, 225] HU. BR-0.5 results maintain good detail, while TMGAN produces target texture which is more uniform and pleasing compared to other methods. With blending, we preserve detail (red arrows) from BR-0.5 and target texture in TMGAN.}
\label{6653real1}
\end{minipage}
\end{figure*}

\begin{figure*}[t!]
\begin{minipage}{\textwidth} 
    \begin{tabular}{c@{}c@{}c@{}c@{}c@{}c}
    Noisy input & MSE-denoiser & BR-$0.5$&WGAN-VGG&TMGAN&TMGAN-blended\\
                \begin{subfigure}[b]{3cm}
            \centerline{\includegraphics[width=3cm]{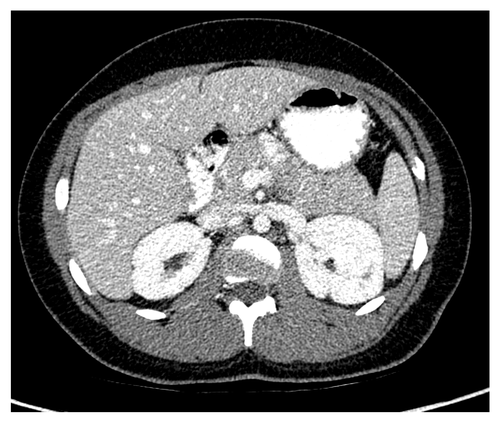}}
            \subcaption{}
            \label{in3}
        \end{subfigure}
     &%
        \begin{subfigure}[b]{3cm}
            \centerline{\includegraphics[width=3cm]{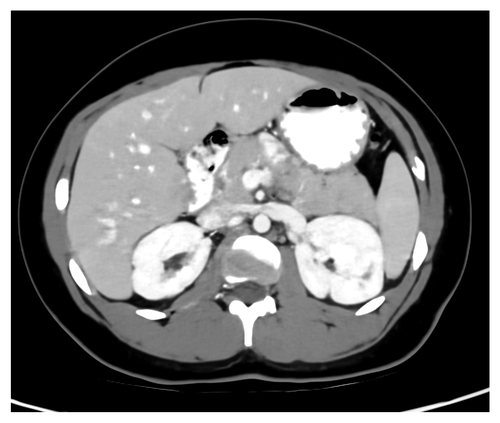}}
            \subcaption{}
            \label{mse3}
        \end{subfigure}
    &%
         \begin{subfigure}[b]{3cm}
             \centerline{\includegraphics[width=3cm]{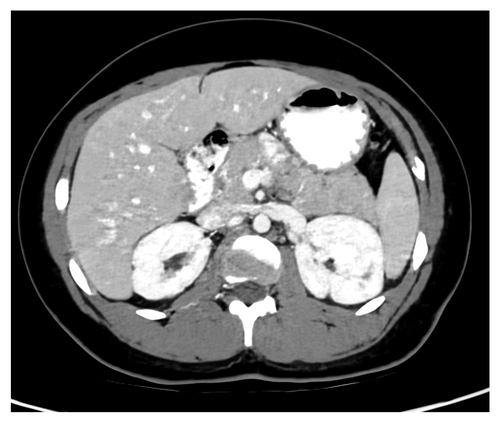}}
             \subcaption{}
             \label{br3}
         \end{subfigure}
     &%
        \begin{subfigure}[b]{3cm}
            \centerline{\includegraphics[width=3cm]{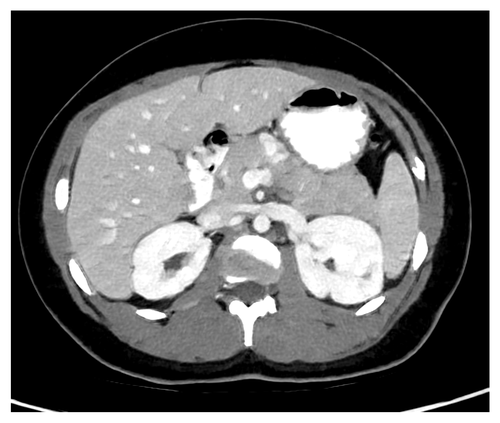}}
            \subcaption{}
            \label{wgan3}
        \end{subfigure}&
        \begin{subfigure}[b]{3cm}
            \centerline{\includegraphics[width=3cm]{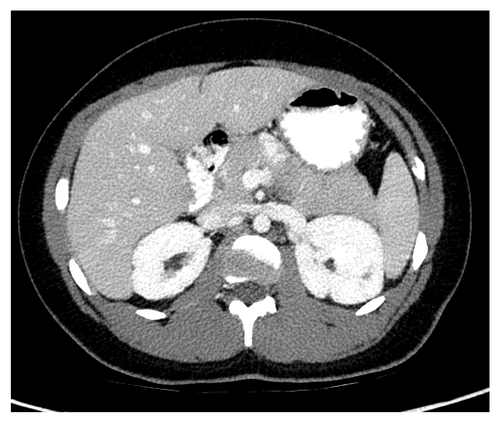}}
            \subcaption{}
            \label{tmgan3}
        \end{subfigure}&
        \begin{subfigure}[b]{3cm}
            \centerline{\includegraphics[width=3cm]{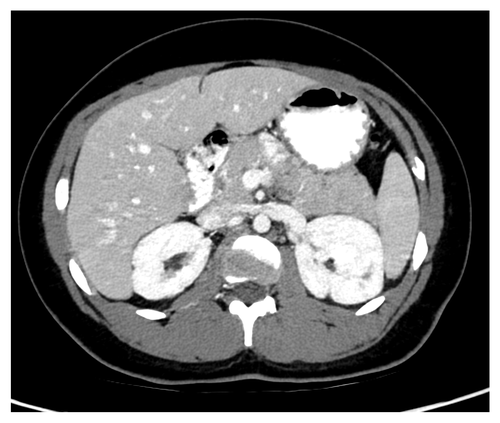}}
            \subcaption{}
            \label{tmganbl3}
        \end{subfigure}
        \\
                \begin{subfigure}[b]{3cm}
            \centerline{\includegraphics[width=3cm]{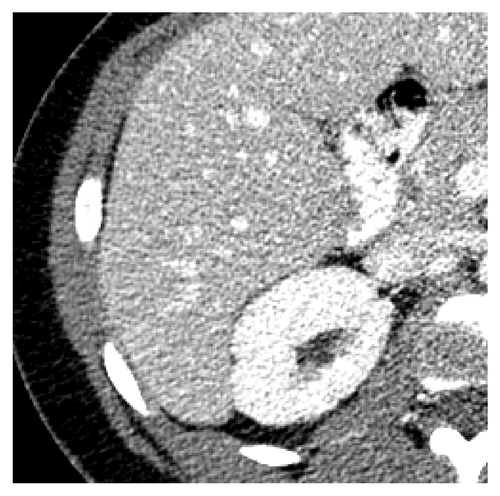}}
            \subcaption{}
            \label{in4}
        \end{subfigure}
     &%
        \begin{subfigure}[b]{3cm}
            \centerline{\includegraphics[width=3cm]{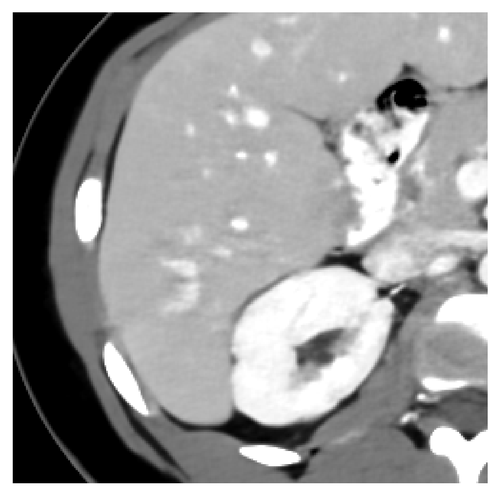}}
            \subcaption{}
            \label{mse4}
        \end{subfigure}
    &%
         \begin{subfigure}[b]{3cm}
             \centerline{\includegraphics[width=3cm]{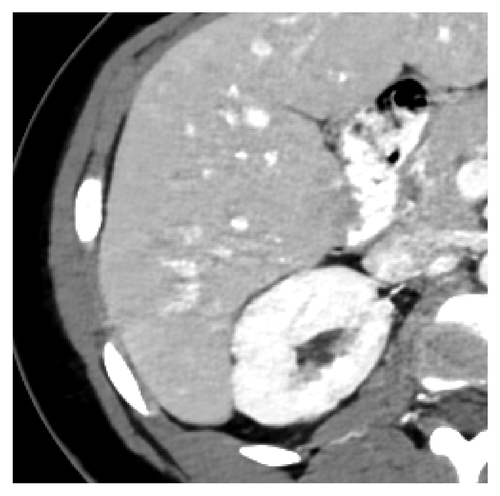}}
             \subcaption{}
             \label{br4}
         \end{subfigure}
     &%
        \begin{subfigure}[b]{3cm}
            \centerline{\includegraphics[width=3cm]{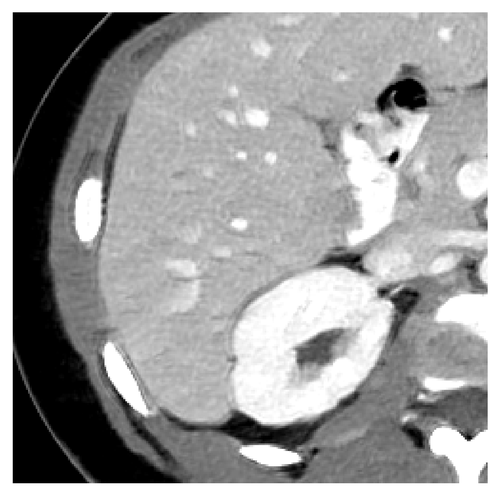}}
            \subcaption{}
            \label{wgan4}
        \end{subfigure}&
        \begin{subfigure}[b]{3cm}
            \centerline{\includegraphics[width=3cm]{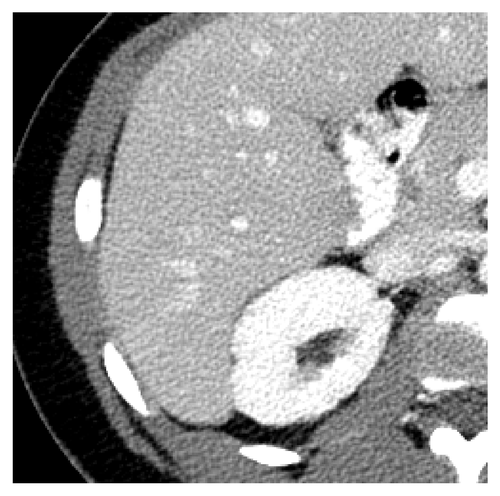}}
            \subcaption{}
            \label{tmgan4}
        \end{subfigure}&
        \begin{subfigure}[b]{3cm}
            \centerline{\includegraphics[width=3cm]{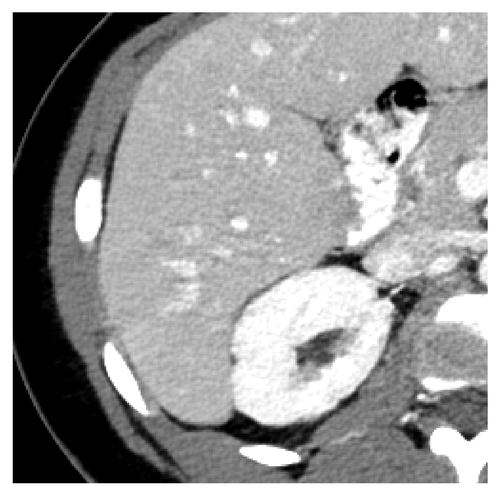}}
            \subcaption{}
            \label{tmganbl4}
        \end{subfigure}
\end{tabular}
\caption{Comparison of slice 2 of denoised results for Exam 1, a low-dose clinical scan. First row shows full slice. Second row shows zoomed ROI. TMGAN trained with $\lambda = 0.4, \sigma=7.8$ HU. Display window is [-125, 225] HU. TMGAN-blended produces target texture for a challenging input with very nonuniform texture.}
\label{6653real2}
\end{minipage}
\end{figure*}
\begin{figure*}[t!]
\begin{minipage}{\textwidth} 
    \begin{tabular}{c@{}c@{}c@{}c@{}c@{}c}
    Noisy input & MSE-denoiser & BR-$0.5$&WGAN-VGG&TMGAN&TMGAN-blended\\
                \begin{subfigure}[b]{3cm}
            \centerline{\includegraphics[width=3cm]{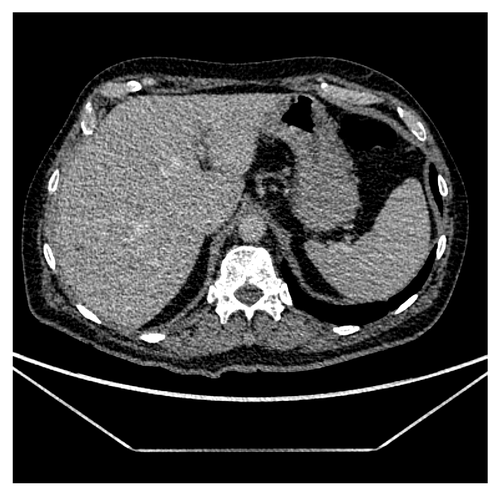}}
            \subcaption{}
            \label{in5}
        \end{subfigure}
     &%
        \begin{subfigure}[b]{3cm}
            \centerline{\includegraphics[width=3cm]{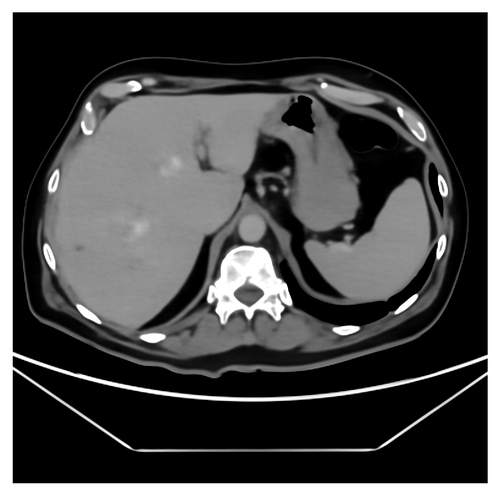}}
            \subcaption{}
            \label{mse5}
        \end{subfigure}
    &%
         \begin{subfigure}[b]{3cm}
             \centerline{\includegraphics[width=3cm]{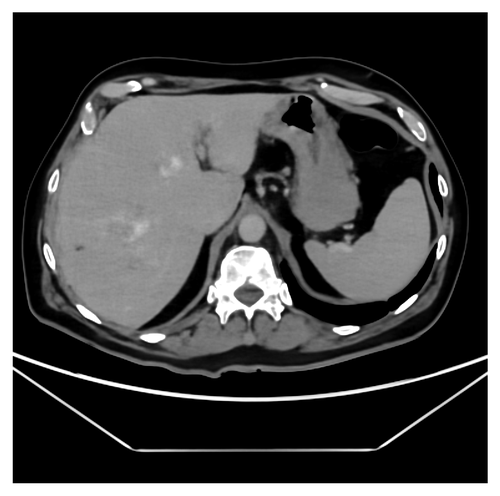}}
             \subcaption{}
             \label{br5}
         \end{subfigure}
     &%
        \begin{subfigure}[b]{3cm}
            \centerline{\includegraphics[width=3cm]{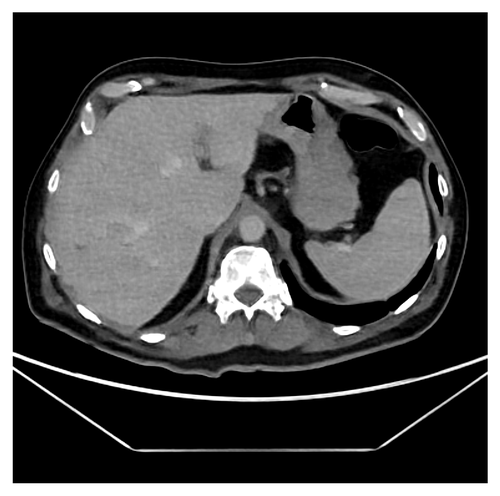}}
            \subcaption{}
            \label{wgan5}
        \end{subfigure}&
        \begin{subfigure}[b]{3cm}
            \centerline{\includegraphics[width=3cm]{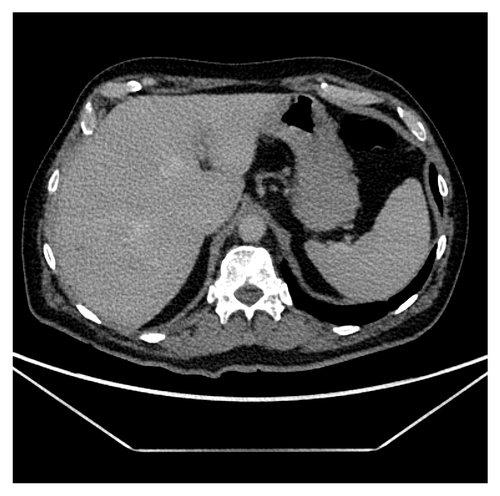}}
            \subcaption{}
            \label{tmgan5}
        \end{subfigure}&
        \begin{subfigure}[b]{3cm}
            \centerline{\includegraphics[width=3cm]{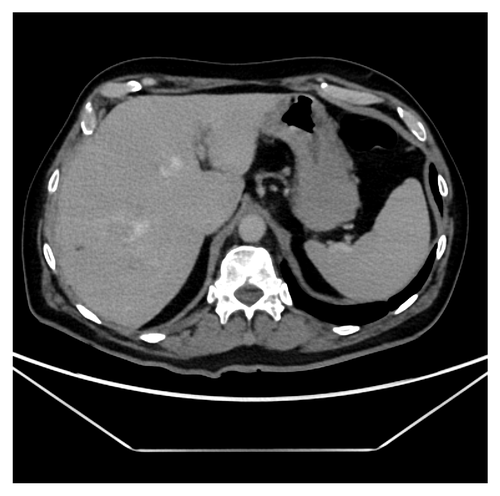}}
            \subcaption{}
            \label{tmganbl5}
        \end{subfigure}
        \\
                \begin{subfigure}[b]{3cm}
            \centerline{\includegraphics[width=3cm]{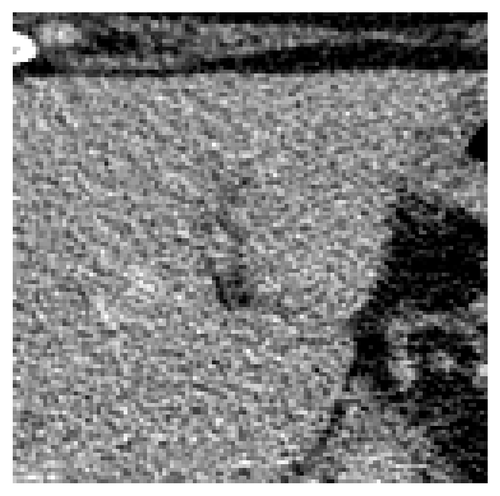}}
            \subcaption{}
            \label{in6}
        \end{subfigure}
     &%
        \begin{subfigure}[b]{3cm}
            \centerline{\includegraphics[width=3cm]{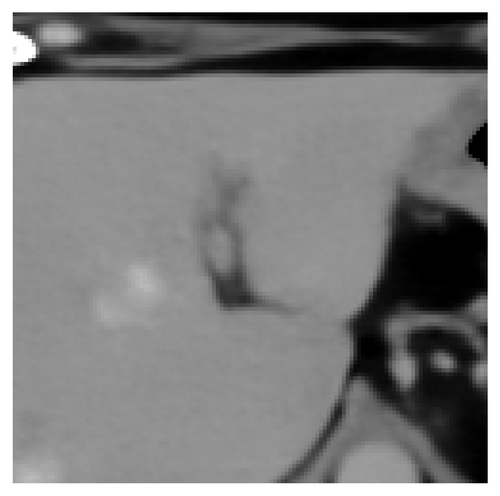}}
            \subcaption{}
            \label{mse6}
        \end{subfigure}
    &%
         \begin{subfigure}[b]{3cm}
             \centerline{\includegraphics[width=3cm]{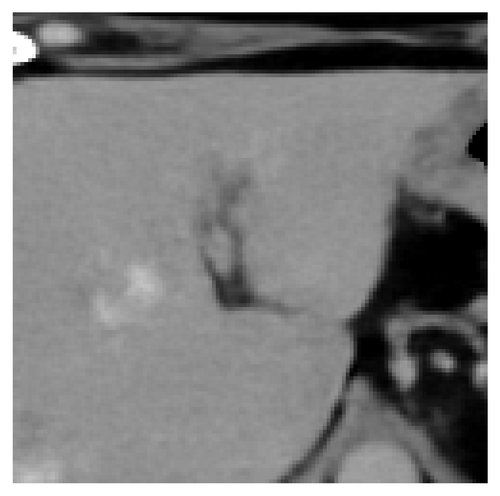}}
             \subcaption{}
             \label{br6}
         \end{subfigure}
     &%
        \begin{subfigure}[b]{3cm}
            \centerline{\includegraphics[width=3cm]{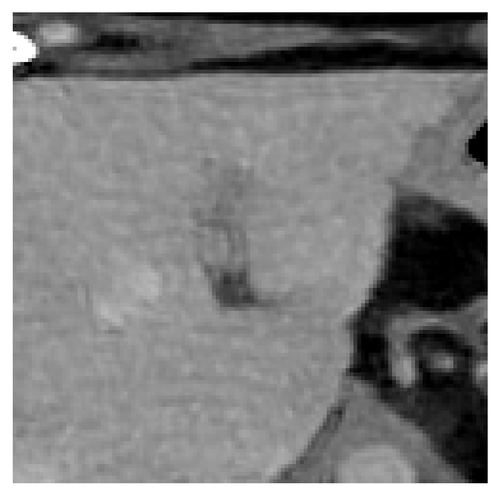}}
            \subcaption{}
            \label{wgan6}
        \end{subfigure}&
        \begin{subfigure}[b]{3cm}
            \centerline{\includegraphics[width=3cm]{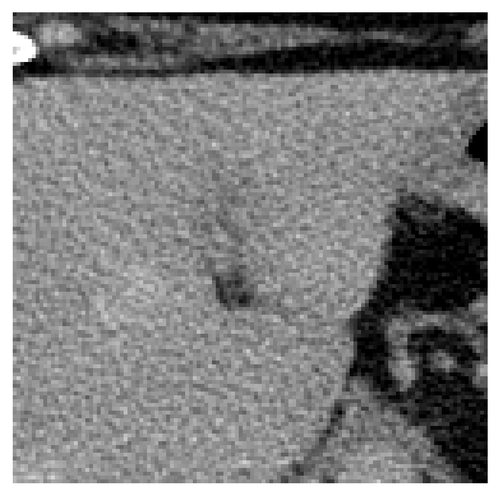}}
            \subcaption{}
            \label{tmgan6}
        \end{subfigure}&
        \begin{subfigure}[b]{3cm}
            \centerline{\includegraphics[width=3cm]{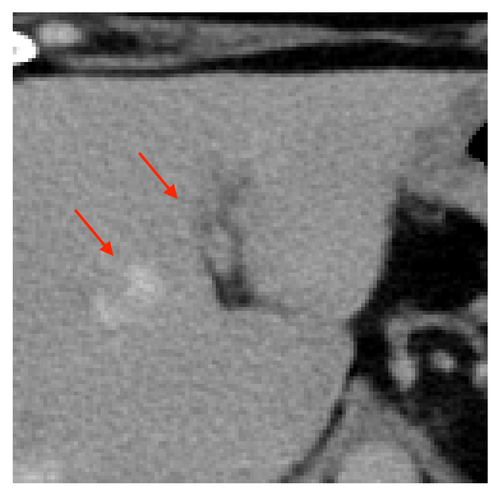}}
            \subcaption{}
            \label{tmganbl6}
        \end{subfigure}
\end{tabular}
\caption{Comparison of slice 1 of denoised results for Exam 2, a low-dose low-contrast clinical scan. First row shows full slice. Second row shows zoomed ROI. TMGAN trained with $\lambda = 0.4, \sigma=7.8$ HU. Display window is [-125, 225] HU. BR-0.5 results maintain good detail, while TMGAN produces target texture which is more uniform and pleasing compared to other methods. With blending, we preserve detail (red arrows) from BR-0.5 and target texture in TMGAN.}
\label{9323real1}
\end{minipage}
\end{figure*}
\begin{figure*}[t!]
\begin{minipage}{\textwidth} 
    \begin{tabular}{c@{}c@{}c@{}c@{}c@{}c}
    Noisy input & MSE-denoiser & BR-$0.5$&WGAN-VGG&TMGAN&TMGAN-blended\\
                \begin{subfigure}[b]{3cm}
            \centerline{\includegraphics[width=3cm]{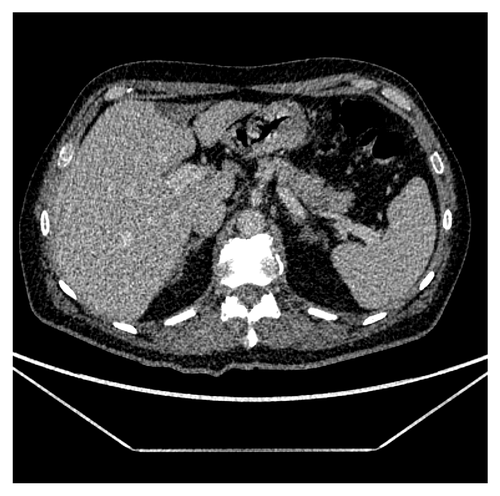}}
            \subcaption{}
            \label{in7}
        \end{subfigure}
     &%
        \begin{subfigure}[b]{3cm}
            \centerline{\includegraphics[width=3cm]{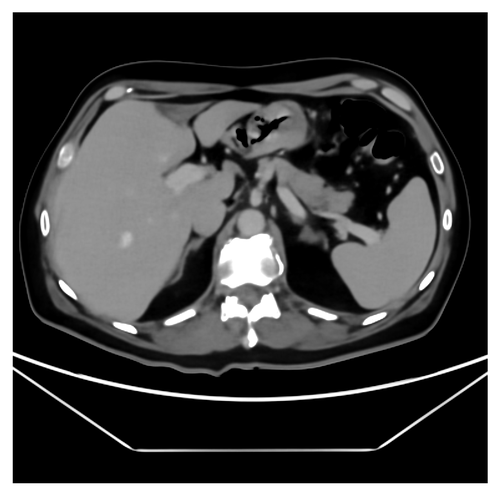}}
            \subcaption{}
            \label{mse7}
        \end{subfigure}
    &%
         \begin{subfigure}[b]{3cm}
             \centerline{\includegraphics[width=3cm]{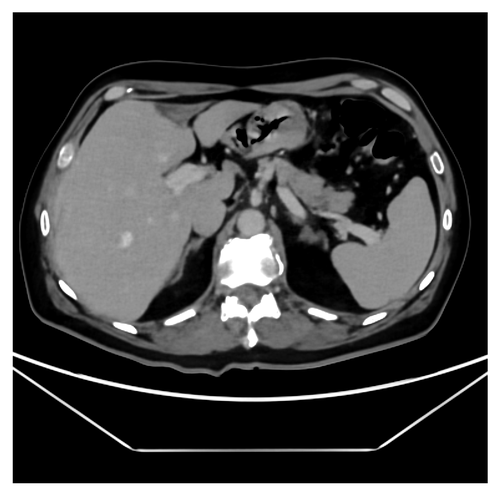}}
             \subcaption{}
             \label{br7}
         \end{subfigure}
     &%
        \begin{subfigure}[b]{3cm}
            \centerline{\includegraphics[width=3cm]{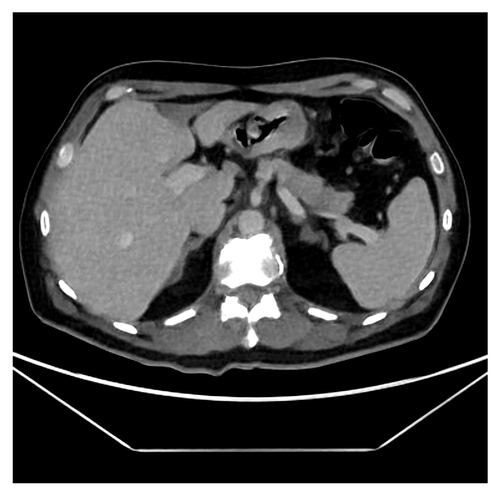}}
            \subcaption{}
            \label{wgan7}
        \end{subfigure}&
        \begin{subfigure}[b]{3cm}
            \centerline{\includegraphics[width=3cm]{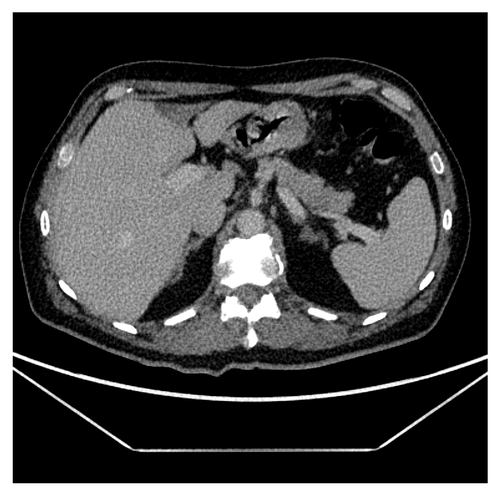}}
            \subcaption{}
            \label{tmgan7}
        \end{subfigure}&
        \begin{subfigure}[b]{3cm}
            \centerline{\includegraphics[width=3cm]{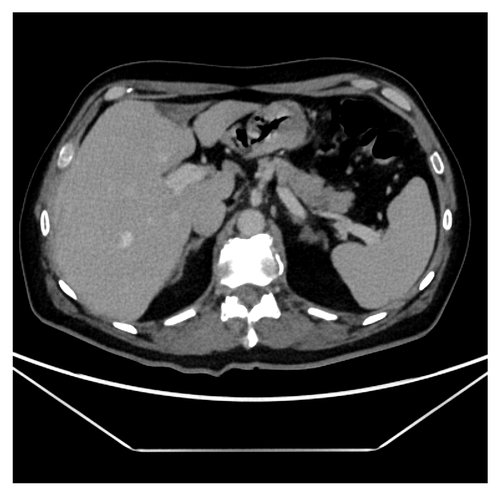}}
            \subcaption{}
            \label{tmganbl7}
        \end{subfigure}
        \\
                \begin{subfigure}[b]{3cm}
            \centerline{\includegraphics[width=3cm]{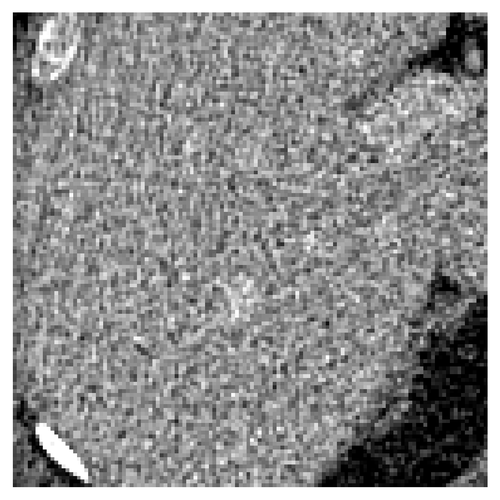}}
            \subcaption{}
            \label{in8}
        \end{subfigure}
     &%
        \begin{subfigure}[b]{3cm}
            \centerline{\includegraphics[width=3cm]{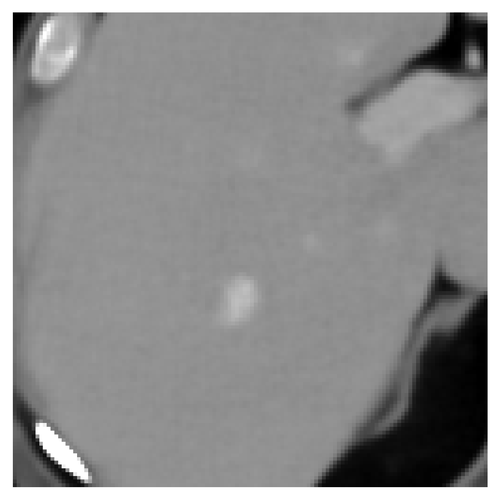}}
            \subcaption{}
            \label{mse8}
        \end{subfigure}
    &%
         \begin{subfigure}[b]{3cm}
             \centerline{\includegraphics[width=3cm]{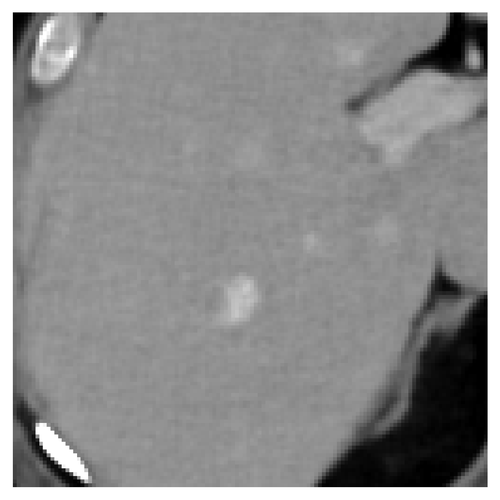}}
             \subcaption{}
             \label{br8}
         \end{subfigure}
     &%
        \begin{subfigure}[b]{3cm}
            \centerline{\includegraphics[width=3cm]{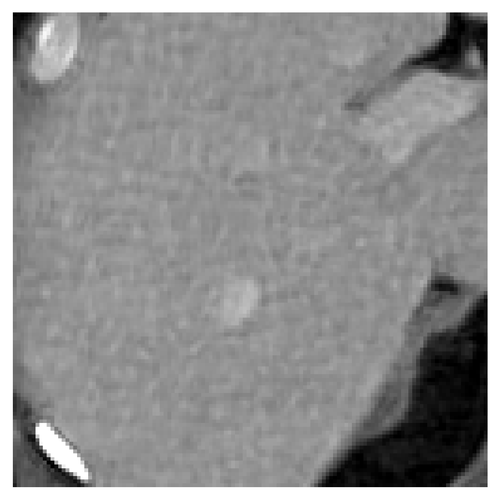}}
            \subcaption{}
            \label{wgan8}
        \end{subfigure}&
        \begin{subfigure}[b]{3cm}
            \centerline{\includegraphics[width=3cm]{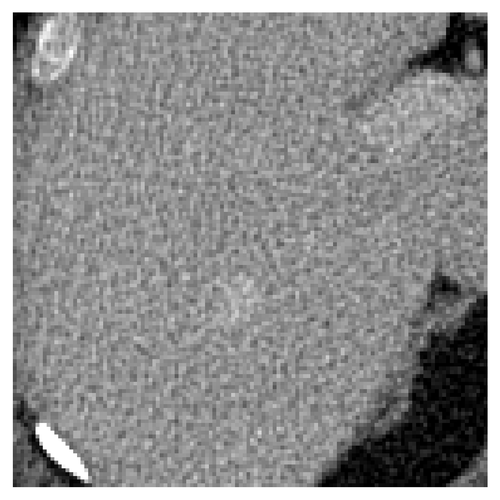}}
            \subcaption{}
            \label{tmgan8}
        \end{subfigure}&
        \begin{subfigure}[b]{3cm}
            \centerline{\includegraphics[width=3cm]{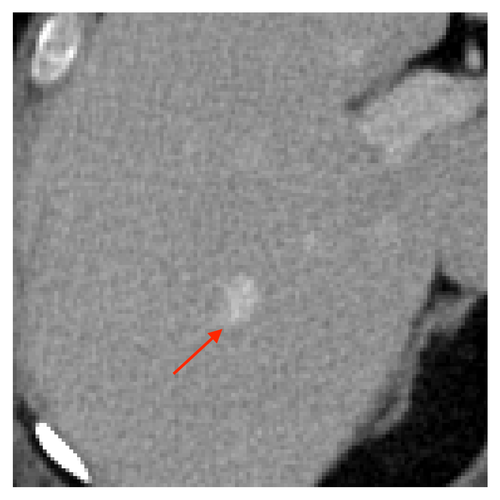}}
            \subcaption{}
            \label{tmganbl8}
        \end{subfigure}
\end{tabular}
\caption{Comparison of slice 2 of denoised results for Exam 2, a low-dose low-contrast clinical scan. First row shows full slice. Second row shows zoomed ROI. TMGAN trained with $\lambda = 0.4, \sigma=7.8$ HU. Display window is [-125, 225] HU. BR-0.5 results maintain good detail, while TMGAN produces target texture which is more uniform and pleasing compared to other methods. With blending, we preserve detail (red arrow) from BR-0.5 and target texture in TMGAN.}
\label{9323real2}
\end{minipage}
\end{figure*}

\fi

\subsubsection{Test Data and Evaluation Metrics}

For quantitative evaluation, we generated  realistic synthesized clean images by reconstructing 9 clinical scans with the GE’s TrueFidelity DLIR option \cite{GEwhitepaper}.
The X-ray tube voltage and current varied from scan to scan in the range of 80-120 kVp and 55-375 mA, respectively. 
We also scanned a water phantom with a tube voltage of 120 kVp and current 350 mA.
The clean volumes were used as ground truth, and noise from the water phantom was added to simulate the scanner noise. None of these scans were used in training. 
Using the ground truth, we computed the peak signal-to-noise ratio (PSNR) and SSIM (structural similarity) \cite{SSIM} metrics. 
We also show NPS (noise
power spectrum) results computed using the methods of~\cite{solomon2012quantitative}.

Table~\ref{testdata} lists the clinical exams used to test the algorithms.
None of these exams were used in training.
For each exam, the table lists the scan content along with various scan parameters.
Exams~1-4 were used for the denoising experiment.  
Exam~5 was captured with an extra large (XL) focal spot size, which produced blurred features in the captured image;  this exam was used for the sharpening experiment.
Since these are clinical exams, no ground truth is available, so we provide qualitative evaluation only. 

For denoising, TMGAN was compared to the following alternatives:
\begin{itemize}
\item \textbf{MSE-denoiser}: Denoiser trained only with MSE loss
\item \textbf{BR-}$0.5$: Denoiser using bias reducing loss function with $\alpha=0.5$ \cite{BRpaper}
\item \textbf{WGAN-VGG}: Method in \cite{Yang2018a} implemented at \cite{GANVGGGit}  
\end{itemize}

For sharpening, TMGAN was compared to the following alternatives:
\begin{itemize}
\item \textbf{MSE-sharpener}: Sharpener trained with MSE loss 
\item \textbf{NPSF}$_1$: NPSF sharpener \cite{NPSFPaper} tuned to maintain the same level of noise as input
\item \textbf{NPSF}$_2$: NPSF sharpener \cite{NPSFPaper} tuned to achieve the same level of noise as TMGAN.
\end{itemize}
For a fair comparison, all denoisers and sharpeners, except WGAN-VGG, used the TM generator architecture.

\ifx\addfigs\undefined

\begin{figure*}[t!]
\begin{minipage}{\textwidth} 
    \begin{tabular}{c@{}c@{}c@{}c@{}c@{}c}
    Noisy input & MSE-denoiser & BR-$0.5$&WGAN-VGG&TMGAN&TMGAN-blended\\
                \begin{subfigure}[b]{3cm}
            \centerline{\includegraphics[width=3cm]{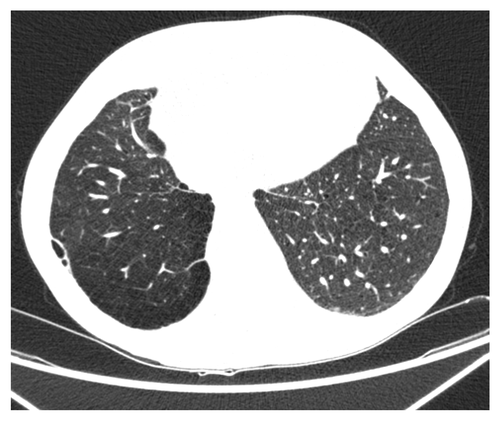}}
            \subcaption{}
            \label{in9}
        \end{subfigure}
     &%
        \begin{subfigure}[b]{3cm}
            \centerline{\includegraphics[width=3cm]{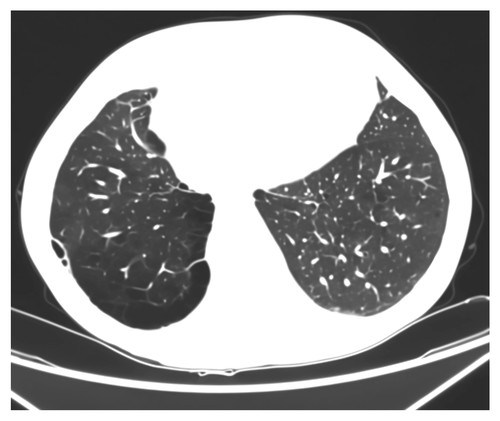}}
            \subcaption{}
            \label{mse9}
        \end{subfigure}
    &%
         \begin{subfigure}[b]{3cm}
             \centerline{\includegraphics[width=3cm]{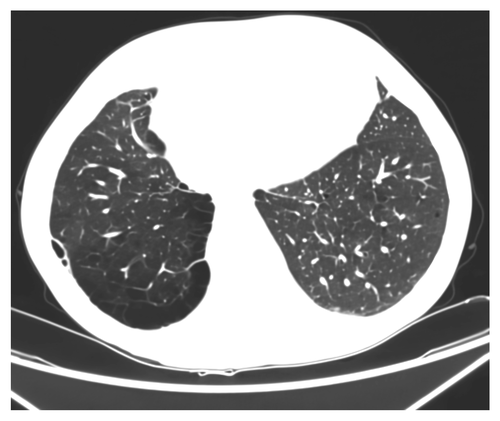}}
             \subcaption{}
             \label{br9}
         \end{subfigure}
     &%
        \begin{subfigure}[b]{3cm}
            \centerline{\includegraphics[width=3cm]{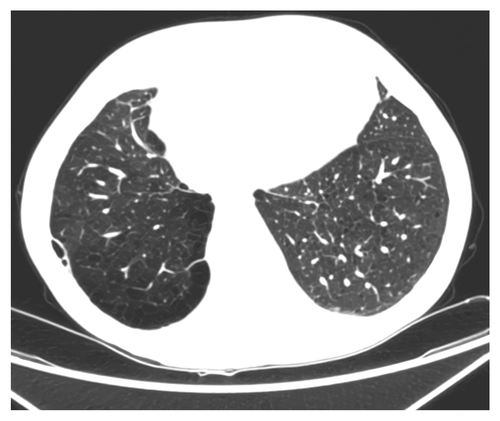}}
            \subcaption{}
            \label{wgan9}
        \end{subfigure}&
        \begin{subfigure}[b]{3cm}
            \centerline{\includegraphics[width=3cm]{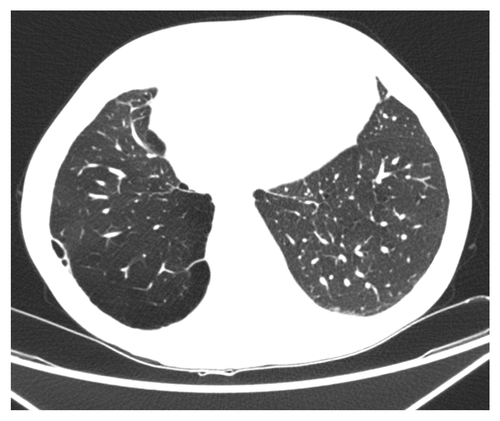}}
            \subcaption{}
            \label{tmgan9}
        \end{subfigure}&
        \begin{subfigure}[b]{3cm}
            \centerline{\includegraphics[width=3cm]{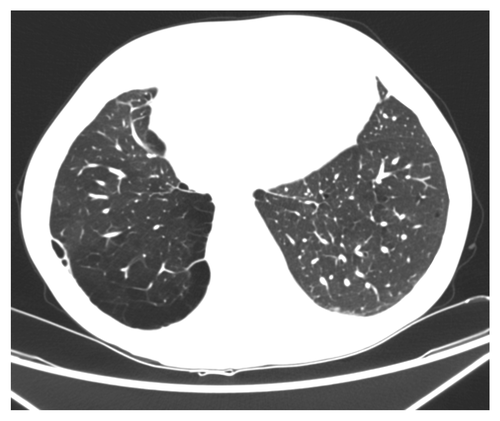}}
            \subcaption{}
            \label{tmganbl9}
        \end{subfigure}
        \\
                \begin{subfigure}[b]{3cm}
            \centerline{\includegraphics[width=3cm]{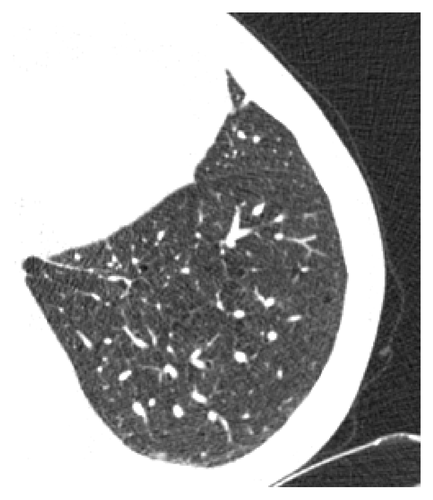}}
            \subcaption{}
            \label{in10}
        \end{subfigure}
     &%
        \begin{subfigure}[b]{3cm}
            \centerline{\includegraphics[width=3cm]{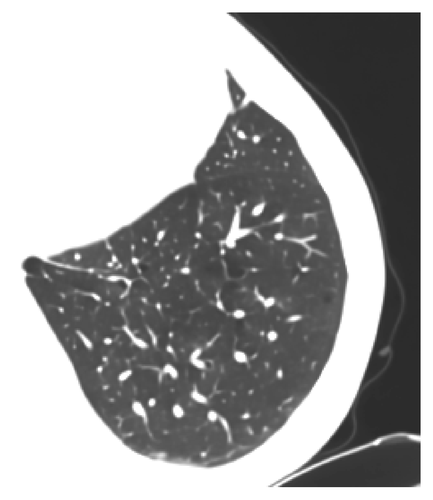}}
            \subcaption{}
            \label{mse10}
        \end{subfigure}
    &%
         \begin{subfigure}[b]{3cm}
             \centerline{\includegraphics[width=3cm]{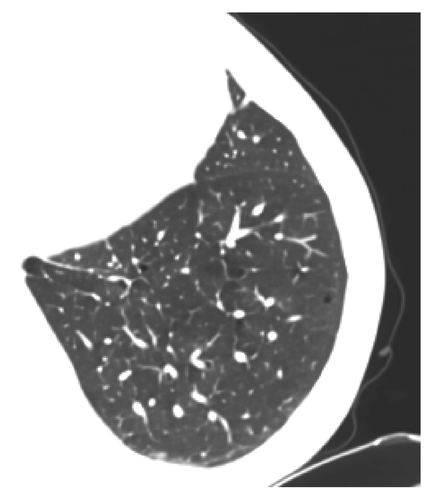}}
             \subcaption{}
             \label{br10}
         \end{subfigure}
     &%
        \begin{subfigure}[b]{3cm}
            \centerline{\includegraphics[width=3cm]{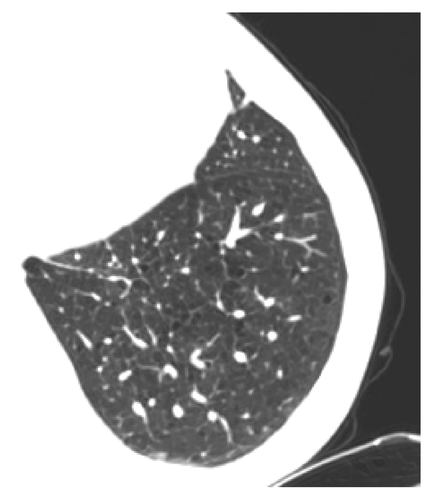}}
            \subcaption{}
            \label{wgan10}
        \end{subfigure}&
        \begin{subfigure}[b]{3cm}
            \centerline{\includegraphics[width=3cm]{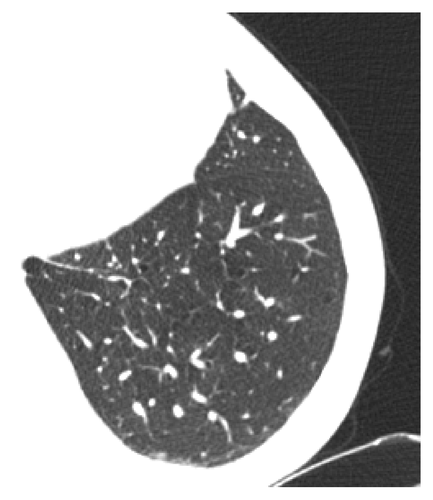}}
            \subcaption{}
            \label{tmgan10}
        \end{subfigure}&
        \begin{subfigure}[b]{3cm}
            \centerline{\includegraphics[width=3cm]{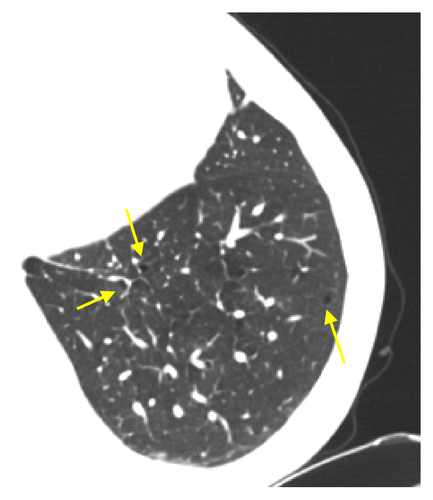}}
            \subcaption{}
            \label{tmganbl10}
        \end{subfigure}
\end{tabular}
\caption{Comparison of denoised results for Exam 3, a low-dose clinical lung scan. First row shows full slice. Second row shows zoomed ROI. TMGAN trained with $\lambda = 0.4, \sigma=7.8$ HU. Display window is [-1200, -200] HU. BR-0.5 results maintain good detail, while TMGAN produces target texture which is more uniform and pleasing compared to other methods. With blending, we preserve detail (small air pockets indicated by yellow arrows) from BR-0.5 and nice texture in TMGAN.}
\label{15107real1}
\end{minipage}
\end{figure*}

\begin{figure*}[t!]
\begin{minipage}{\textwidth} 
    \begin{tabular}{c@{}c@{}c@{}c@{}c@{}c}
    Noisy input & MSE-denoiser & BR-$0.5$&WGAN-VGG&TMGAN&TMGAN-blended\\
                \begin{subfigure}[b]{3cm}
            \centerline{\includegraphics[width=3cm]{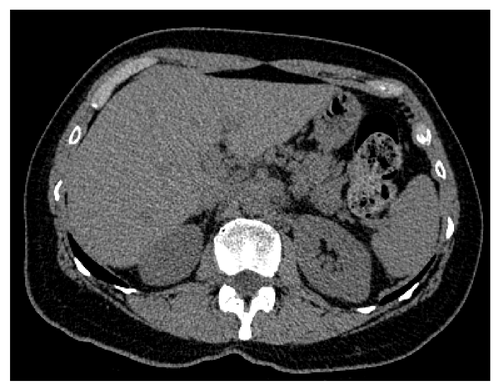}}
            \subcaption{}
            \label{in11}
        \end{subfigure}
     &%
        \begin{subfigure}[b]{3cm}
            \centerline{\includegraphics[width=3cm]{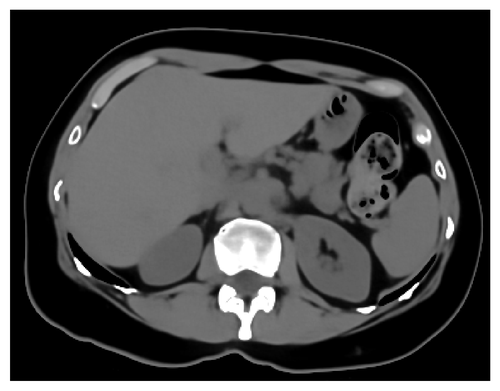}}
            \subcaption{}
            \label{mse11}
        \end{subfigure}
    &%
         \begin{subfigure}[b]{3cm}
             \centerline{\includegraphics[width=3cm]{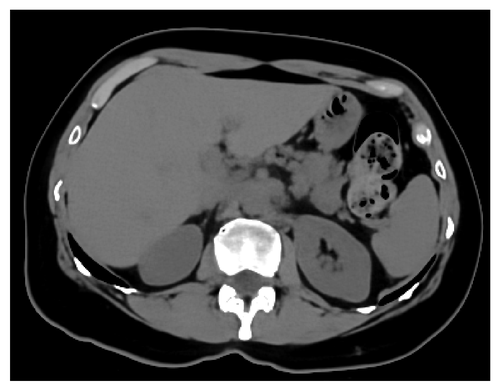}}
             \subcaption{}
             \label{br11}
         \end{subfigure}
     &%
        \begin{subfigure}[b]{3cm}
            \centerline{\includegraphics[width=3cm]{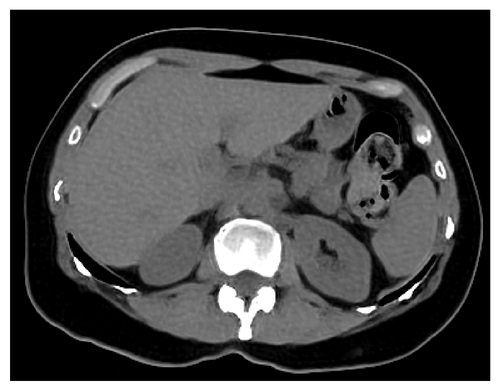}}
            \subcaption{}
            \label{wgan11}
        \end{subfigure}&
        \begin{subfigure}[b]{3cm}
            \centerline{\includegraphics[width=3cm]{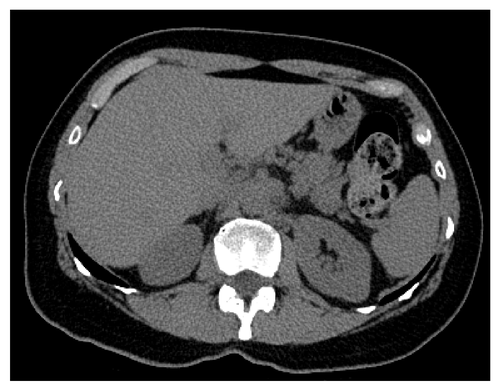}}
            \subcaption{}
            \label{tmgan11}
        \end{subfigure}&
        \begin{subfigure}[b]{3cm}
            \centerline{\includegraphics[width=3cm]{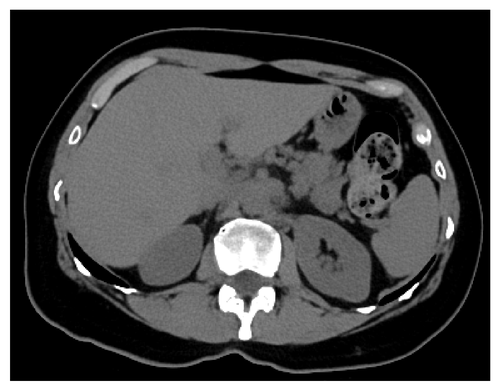}}
            \subcaption{}
            \label{tmganbl11}
        \end{subfigure}
        \\
                \begin{subfigure}[b]{3cm}
            \centerline{\includegraphics[width=3cm]{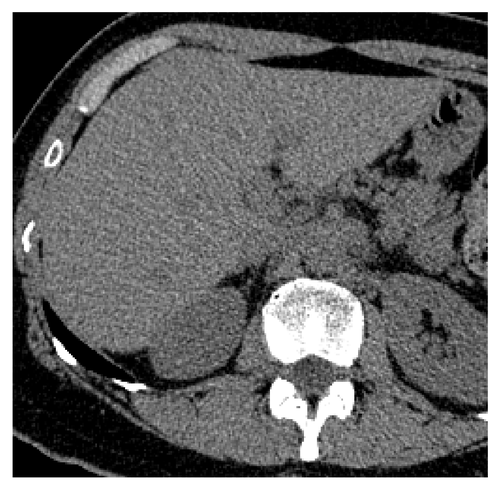}}
            \subcaption{}
            \label{in12}
        \end{subfigure}
     &%
        \begin{subfigure}[b]{3cm}
            \centerline{\includegraphics[width=3cm]{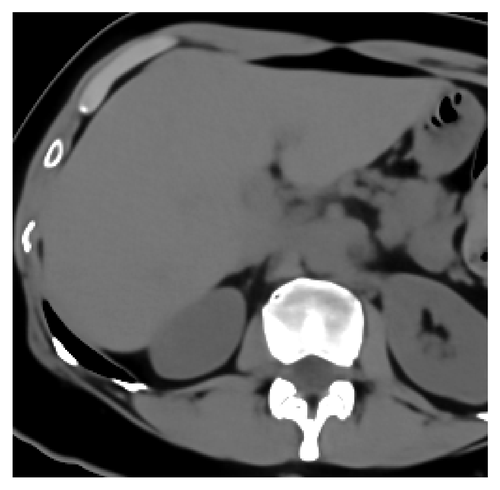}}
            \subcaption{}
            \label{mse12}
        \end{subfigure}
    &%
         \begin{subfigure}[b]{3cm}
             \centerline{\includegraphics[width=3cm]{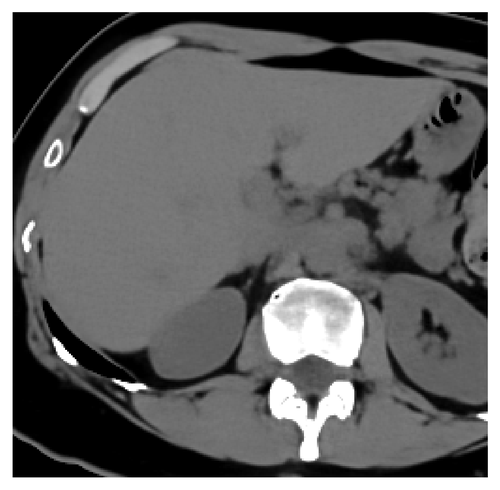}}
             \subcaption{}
             \label{br12}
         \end{subfigure}
     &%
        \begin{subfigure}[b]{3cm}
            \centerline{\includegraphics[width=3cm]{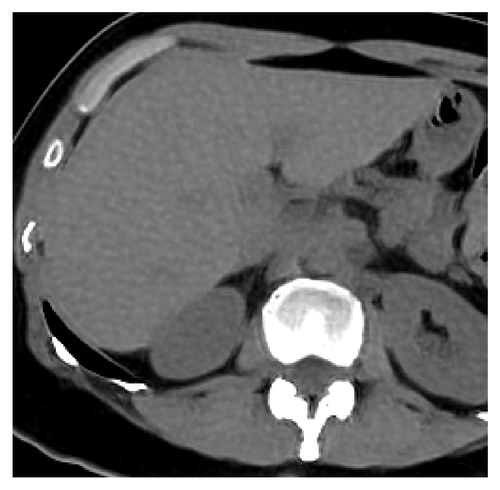}}
            \subcaption{}
            \label{wgan12}
        \end{subfigure}&
        \begin{subfigure}[b]{3cm}
            \centerline{\includegraphics[width=3cm]{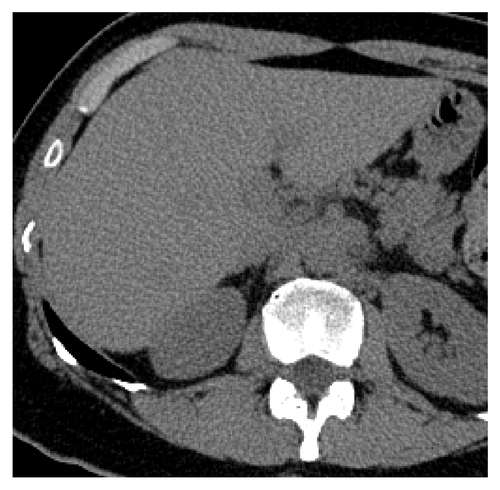}}
            \subcaption{}
            \label{tmgan12}
        \end{subfigure}&
        \begin{subfigure}[b]{3cm}
            \centerline{\includegraphics[width=3cm]{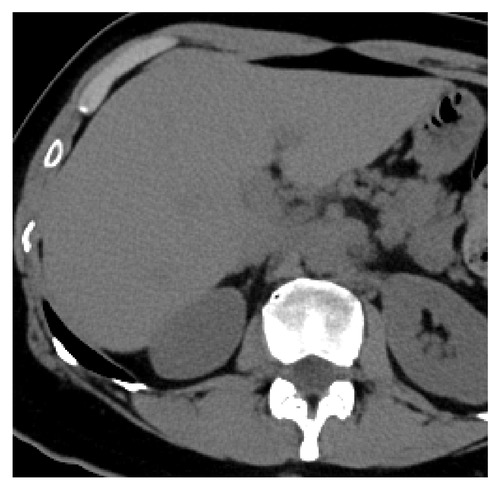}}
            \subcaption{}
            \label{tmganbl12}
        \end{subfigure}
\end{tabular}
\caption{Comparison of denoised results for Exam 4, a low-dose clinical scan. First row shows full slice. Second row shows zoomed ROI. TMGAN trained with $\lambda = 0.4, \sigma=7.8$ HU. Display window is [-125, -225] HU. TMGAN-blended recovers good detail and a uniform target texture given a challenging input.}
\label{15461real1}
\end{minipage}
\end{figure*}
\fi

\section{Results and Discussion}\label{allresults}

\subsection{Quantitative evaluation}
Fig.~\ref{figWP_example} illustrates the textures generated by TMGAN with an input image of the test water phantom using the bone+ kernel and a target texture of a water phantom using the standard kernel.
Both phantoms are reconstructed with a DFOV of 15 cm.
Fig.~\ref{figWP_example} (\subref{lam0}-\subref{lam004}) show the output of TMGAN with $\lambda$ values of 0.0, 0.01, and 0.04, respectively, while $\sigma$ was set to $50$ HU. 
All images use a [-175, 175] HU window and a 5.86 cm FOV. 
Fig.~\ref{fignps1d} shows the corresponding NPS plots for all 5 images.

\ifx\addfigs\undefined

\begin{figure*}[t!]
\begin{minipage}{0.6\textwidth} 
    \begin{tabular}{c@{}c@{}c@{}c@{}c}
     \begin{tabular}{@{}c@{}}Blurred noisy \\ input\end{tabular} &
     \begin{tabular}{@{}c@{}}MSE-\\sharpener\end{tabular} &
      \begin{tabular}{@{}c@{}}NPSF$_1$\end{tabular} &
      \begin{tabular}{@{}c@{}}NPSF$_2$\end{tabular} &
        TMGAN\\
   
                \begin{subfigure}[b]{2.3cm}
            \centerline{\includegraphics[width=2.4cm]{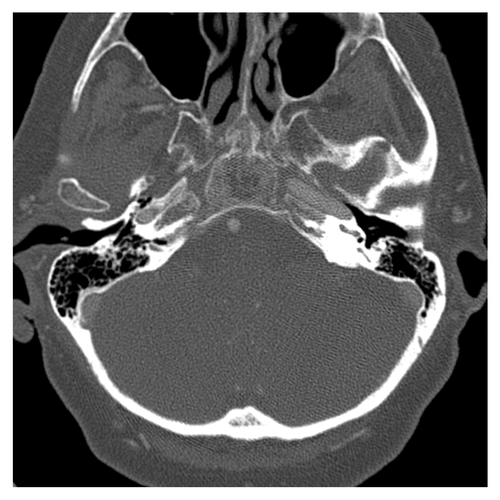}}
            \subcaption{}
            \label{in9bl}
        \end{subfigure}
     &%
        \begin{subfigure}[b]{2.3cm}
            \centerline{\includegraphics[width=2.4cm]{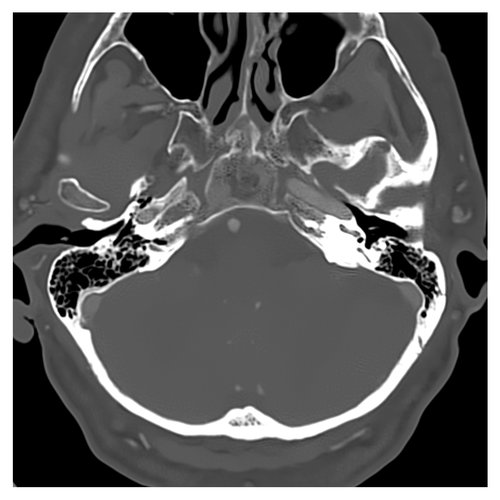}}
            \subcaption{}
            \label{mse9bl}
        \end{subfigure}
    &%
         \begin{subfigure}[b]{2.3cm}
             \centerline{\includegraphics[width=2.4cm]{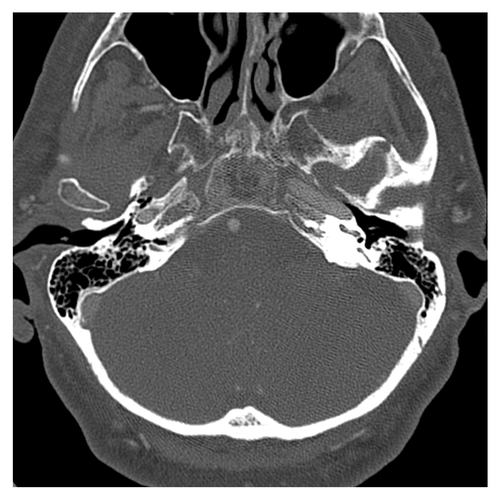}}
             \subcaption{}
             \label{br9bl}
         \end{subfigure}
     &%
        \begin{subfigure}[b]{2.3cm}
            \centerline{\includegraphics[width=2.4cm]{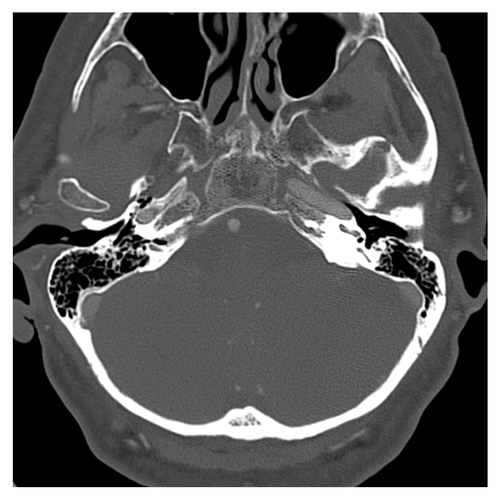}}
            \subcaption{}
            \label{gan9bl}
        \end{subfigure}
             &%
        \begin{subfigure}[b]{2.3cm}
            \centerline{\includegraphics[width=2.4cm]{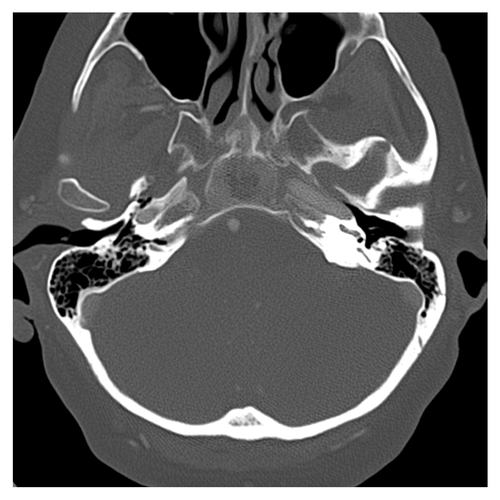}}
            \subcaption{}
            \label{tmgan9bl}
        \end{subfigure}
        \\
                \begin{subfigure}[b]{2.3cm}
            \centerline{\includegraphics[width=2.4cm]{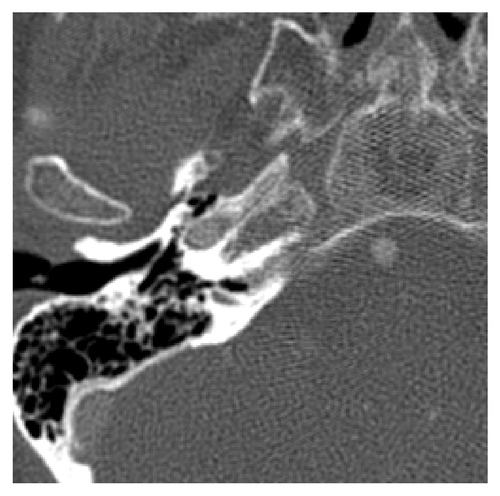}}
            \subcaption{}
            \label{in10bl}
        \end{subfigure}
     &%
        \begin{subfigure}[b]{2.3cm}
            \centerline{\includegraphics[width=2.4cm]{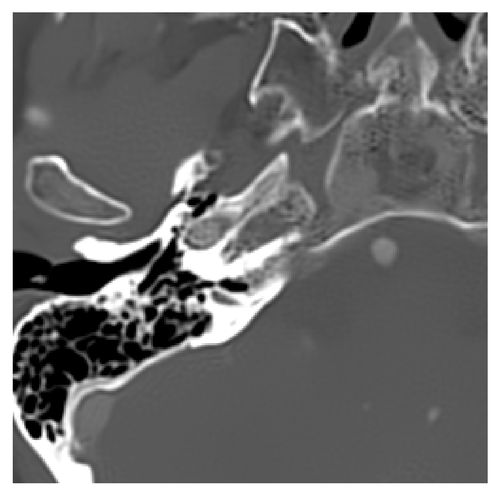}}
            \subcaption{}
            \label{mse10bl}
        \end{subfigure}
    &%
         \begin{subfigure}[b]{2.3cm}
             \centerline{\includegraphics[width=2.4cm]{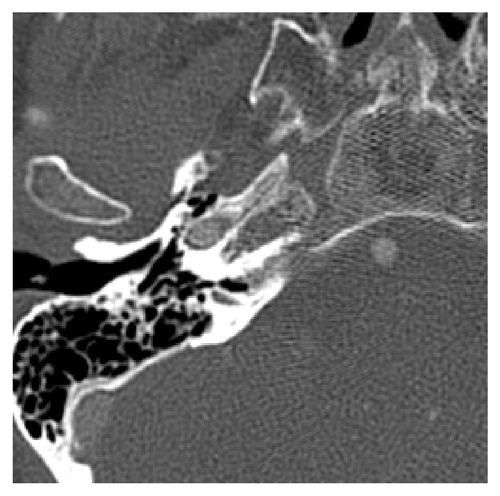}}
             \subcaption{}
             \label{br10bl}
         \end{subfigure}
     &%
        \begin{subfigure}[b]{2.3cm}
            \centerline{\includegraphics[width=2.4cm]{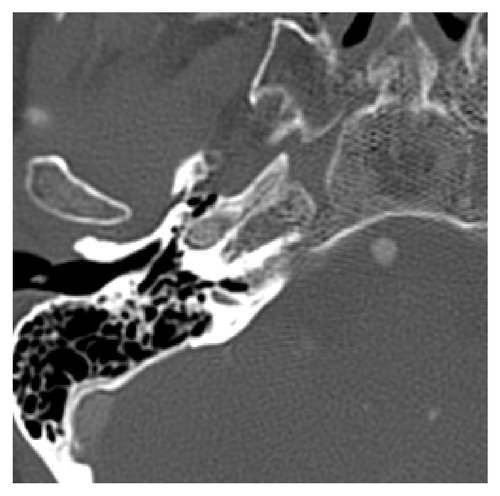}}
            \subcaption{}
            \label{gan10bl}
        \end{subfigure}
             &%
        \begin{subfigure}[b]{2.3cm}
            \centerline{\includegraphics[width=2.4cm]{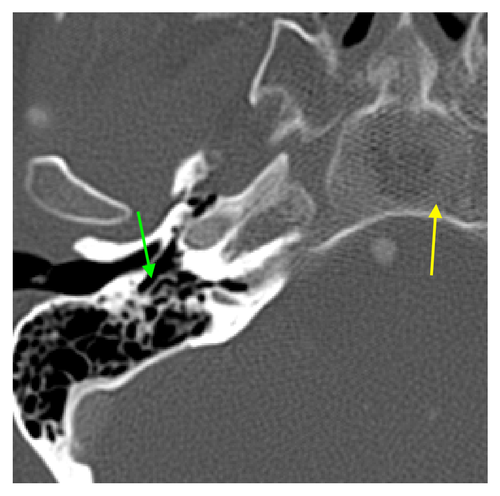}}
            \subcaption{}
            \label{tmgan10bl}
        \end{subfigure}
\end{tabular}
\caption{Comparison of sharpening results for Exam 5, a low-dose and XL focal spot scan. 
First row shows full slice. Second row shows zoomed ROI. TMGAN trained with $\lambda = 0.04, \sigma=50$ HU. 
Display window is [-650, 1350] HU. TMGAN sharpens temporal bones (green arrow), while reducing aliasing artifacts (yellow arrow).}
\label{deblurringReal}
\end{minipage}
\hspace{2em}
\begin{minipage}[t]{0.36\textwidth}
\begin{tabular}{c}
    \centering
    \includegraphics[width=0.98\linewidth]{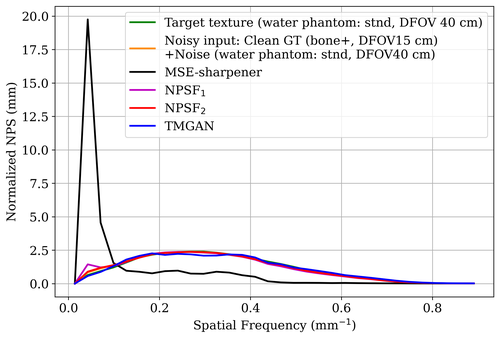}
\end{tabular}
    \vspace{0.8em}\caption{Comparison of NPS for sharpening results. 
    TMGAN and NPSF ($\beta > 0$) match the NPS of results to the target texture.}
    \label{figdeblnps}
\end{minipage}
\end{figure*}

\fi

Consistent with the role of $\lambda$ in controlling the importance of texture mapping in the generator loss function, as $\lambda$ increases, TMGAN generates an output texture that matches more closely to the target texture both qualitatively and in terms of NPS plots.
Moreover, for $\lambda =0 $, the NPS is skewed toward the lower frequencies, which is known to cause an overly-smooth or ``cartoony'' texture in CT images~\cite{dlreview}.

Table~\ref{stdfig3} lists the standard deviation of the noise for each texture shown in Fig.~\ref{figWP_example}. Again consistent with the role of $\lambda$, we see that the noise standard deviation increases as $\lambda$ increases.  This is expected since $\lambda=0$ corresponds to minimizing the bias-reduced MSE only, while increasing $\lambda$ promotes additional texture, which appears as increased noise energy.  However, even with $\lambda=0.04$ the TMGAN noise standard deviation is much less than the input noise standard deviation, indicating that TMGAN is able to simultaneously reduce noise and match texture. 

Fig.~\ref{figpsnr} shows plots of the PSNR of axial slices for each of the denoising algorithms using one of the synthetic noisy exams as an input.
The MSE-denoiser generates the highest PSNR since it minimizes the MSE loss function.
The \mbox{BR-$0.5$} algorithm recovers more detail at the cost of a slight decrease in the PSNR.
Since the GAN architectures of \mbox{WGAN-VGG} and TMGAN have a loss function that encourages texture recovery, they all have lower PSNR than the MSE-denoiser and \mbox{BR-$0.5$} methods.
Finally, PSNR for \mbox{TMGAN-blended} is still lower than \mbox{BR-$0.5$} but is higher than that of \mbox{WGAN-VGG}. 

Table~\ref{tab:quanteval} lists the PSNR and SSIM values averaged over 9 results obtained by inputting synthetic noisy exams to each algorithm.
Notice that \mbox{BR-$0.5$} has the best SSIM value. 
However, as we demonstrate next, the \mbox{TMGAN-blended} results produce much better texture with only a small decrease in the PSNR and SSIM.  

Fig.~\ref{fignpsdenoised} compares the NPS of input and target textures along with the NPS of denoised results using various algorithms (whereas Fig.~\ref{fignps1d} varied $\lambda$). Both input and target textures were obtained by reconstructing the test water phantom with the standard kernel and a DFOV of $40$ cm.
The NPS for TMGAN most closely matches the target texture, while the NPS for TMGAN-blended has slightly increased low frequencies relative to TMGAN. 
More importantly, the NPS for TMGAN-blended matches the NPS of target texture at higher frequencies more closely than all other algorithms except TMGAN.

\subsection{Qualitative evaluation with measured data}

Fig.~\ref{6653real1} and~\ref{6653real2} show the results for two separate slices of Exam 1, a high-contrast clinical scan. 
Visually, results for the MSE-denoiser have smooth texture and lack detail.
In contrast, \mbox{BR-$0.5$} has more detail than MSE-denoiser, but with very nonuniform texture. 
The \mbox{WGAN-VGG} method recovers some texture and detail; however, the texture is not uniform. 

From Fig.~\ref{6653real1}(\subref{tmgan2}) and~\ref{6653real2}(\subref{tmgan4}), TMGAN produces uniform texture, but with increased noise variance and with some details obscured by the texture.
Alternatively, \mbox{TMGAN-blended} achieves the desired uniform texture, along with reduced noise and more visible detail.
More specifically, the arrows in Fig.~\ref{6653real1}(\subref{tmganbl2}) indicate the detail recovered by \mbox{TMGAN-blended} even while maintaining the uniform target texture of TMGAN.

Fig.~\ref{9323real1} and~\ref{9323real2} show results for Exam 2, a low-contrast clinical exam. 
Both the full slice and zoomed views show that TMGAN produces a uniform texture for this low-contrast exam, while \mbox{WGAN-VGG} produces a uniform but coarser texture.
More importantly, the arrows in Fig.~\ref{9323real1}(\subref{tmganbl6}) and~\ref{9323real2}(\subref{tmganbl8}) show that low-contrast features are best detected using the \mbox{TMGAN-blended} results. 

Fig.~\ref{15107real1} shows results for Exam 3 in the lungs. 
Note that the small air pockets in the lungs, shown by the yellow arrows in Fig.~\ref{15107real1}(\subref{tmganbl10}), have diagnostic value. 
These air pockets are not clearly visible in the MSE-denoiser results. 
On the other hand, \mbox{WGAN-VGG}, \mbox{BR-$0.5$}, TMGAN and \mbox{TMGAN-blended} recover them in the denoised images.

Fig.~\ref{15461real1} shows the results of Exam 4, which is a challenging exam due to the very low contrast.
As seen in Fig.~\ref{15461real1}(\subref{tmganbl12}), \mbox{TMGAN-blended} produces the target texture and recovers most of the detail seen in \mbox{WGAN-VGG}, which produces a less desirable texture. 

Fig.~\ref{deblurringReal} shows TMGAN results for sharpening an image consisting of noise and aliasing artifacts. 
From Fig.~\ref{deblurringReal}, it is evident that the MSE-sharpener results are over-smooth and contain artifacts.
While the NPSF$_1$ has more detail and texture, it retains some aliasing artifacts. If we tune NPSF to have noise power the same as the TMGAN results to get NPSF$_2$ then, there is partial noise reduction and sharpening. 
However, the aliasing artifacts look worse in this case. 
In comparison, the TMGAN results are sharper than the input (green arrow) and have more uniform texture with a lower noise level than NPSF$_1$. 
Importantly, the TMGAN results have reduced aliasing artifacts while remaining sharp as indicated by the yellow arrow.

Fig.~\ref{figdeblnps} compares the NPS for a sharpened phantom scan with uniform areas using the algorithms discussed here. 
The NPS for the MSE-sharpener is skewed towards low frequencies producing the over-smooth texture observed in Fig.~\ref{deblurringReal}. 
The NPS for NPSF and TMGAN results are very similar, and they match with the target texture.

\section{Conclusion}
\label{conclusion}

We proposed a novel neural network, TMGAN, that denoises and/or sharpens CT images while simultaneously matching the texture of the resulting output to a target texture.
We achieve this using a branched network with identical weights in each branch.
Each branch processes the ground truth corrupted by noise, with the noise realization independent in the two branches.
By taking the difference of the resulting outputs, our network separates texture from image with anatomical detail.
By embedding this network in an adversarial training framework, we train to produce a desired texture layered on top of a clean image. The resulting output is an enhanced CT image that contains important physiological details and maintains a texture that is viewed as desirable by practicing radiologists.
Our method reduces the risk of hallucination by separating the clean CT image containing anatomical features from texture and restricting generation to the texture part of the image.
Furthermore, the bias-variance trade-off can be modulated as desired by using a simple blending method. Our experiments show that TMGAN removes streaking or aliasing artifacts and produces uniform texture while maintaining important detail.

\section{Acknowledgment}
We thank Jonathan S.~Maltz, Roman Melnyk, Brian Nett, and Ken D.~Sauer for fruitful discussions and Karen Procknow for her valuable clinical feedback.

\bibliographystyle{IEEEbib}
\bibliography{refs}
\end{document}